\documentclass[final,12pt,notitlepage]{article}
\usepackage{amsfonts}
\usepackage{amsmath}
\usepackage{amssymb}
\usepackage{mathtools}
\usepackage{amsthm}
\usepackage[svgnames]{xcolor}
\usepackage[sc,labelsep=period,font=small]{caption}
\usepackage[margin=.96in]{geometry}
\usepackage[bottom,multiple,splitrule]{footmisc}
\usepackage{graphicx}
\usepackage[authoryear,round]{natbib}
\usepackage[inline]{enumitem}
\usepackage[labelfont=rm]{subfig}
\usepackage{hyperref}
\usepackage[onehalfspacing]{setspace}
\usepackage{titlesec}
\usepackage{sectsty}
\usepackage{dsfont}
\usepackage{enumitem}
\usepackage{mathrsfs}
\usepackage{xargs}
\usepackage[normalem]{ulem}
\usepackage{comment}
\usepackage[capitalize,nameinlink]{cleveref}
\usepackage{bm}
\usepackage{tkz-euclide}
\usetikzlibrary{angles}

\DeclareMathOperator*{\argmax}{arg\,max}
\renewcommand{\equiv}{:=}
\renewcommand{\Re}{\mathbb{R}}

\hypersetup{
pdftitle={From Design to Disclosure},
pdfauthor={S. Nageeb Ali, Andreas Kleiner, and Kun Zhang},
pdfkeywords={TBA}
}
\hypersetup{
colorlinks,breaklinks,naturalnames,
citecolor=DarkBlue,urlcolor=DarkBlue,linkcolor=DarkBlue
}

\usepackage{tikz}
\usetikzlibrary{patterns, decorations.pathreplacing, arrows.meta}
\usetikzlibrary{arrows}

\definecolor{webblue}{rgb}{0.2,0.1,0.7}

\definecolor{cred}{RGB}{140,21,21}

\theoremstyle{plain}
\newtheorem{theorem}{Theorem}
\newtheorem{assumption}{Assumption}

\newtheorem{lemma}{Lemma}

\newtheorem{proposition}{Proposition}
\newtheorem{remark}{Remark}
\newtheorem{definition}{Definition}

\theoremstyle{definition}
\theoremstyle{remark}
\newtheorem{example}{Example}

\crefname{assumption}{assumption}{assumptions}
\crefname{claim}{claim}{claims}
\crefname{fact}{fact}{facts}
\sectionfont{\color{DarkRed}}
\subsectionfont{\color{DarkRed}}
\subsubsectionfont{\color{DarkRed}}

\renewcommand{\qedsymbol}{Q.E.D.}

\renewcommand{\qedsymbol}{$\blacksquare$}

\newcommand{\R}{\mathbb{R}}
\newcommand{\E}{\mathbb{E}}

\newcommand{\uv}{\underline{\theta}}
\newcommand{\ov}{\overline{\theta}}

\newcommand{\Q}{\mathcal{Q}}

\DeclareMathOperator{\argmin}{arg\ min}

\DeclareMathOperator{\supp}{supp}
\newcommand{\e}{\varepsilon}

\newcommand{\norm}[1]{ \lVert #1 \rVert }

\usepackage{xcolor}
\usepackage{tikz}
\usetikzlibrary{patterns,hobby,snakes}
\usetikzlibrary{shapes.misc, positioning}

\begin{document}

\begin{titlepage}
\title{
From Design to Disclosure\thanks{We thank Navin Kartik, Stephen Morris, Jacopo Perego, Andy Skrzypacz, Roland Strausz, Kai Hao Yang, and various seminar and conference attendees. Kleiner acknowledges financial support from the German Research Foundation (DFG) through Germany’s Excellence Strategy - EXC 2047/1 - 390685813, EXC 2126/2 – 390838866, and the CRC TR-224 (Project B02).}

\vspace{.2in}
}  

\author{
{ S. Nageeb Ali}\thanks{Department of Economics, Pennsylvania State University. Email: \href{mailto:nageeb@psu.edu}{nageeb@psu.edu}.}
\and 
{Andreas Kleiner}\thanks{Department of Economics, University of Bonn. Email: \href{mailto:andykleiner@gmail.com}{andykleiner@gmail.com}.}
\and 
{Kun Zhang}\thanks{School of Economics, University of Queensland. Email: \href{mailto:kun@kunzhang.org}{kun@kunzhang.org}.}
\vspace{.2in}}

\begingroup
\singlespacing
\maketitle
\endgroup

\begin{abstract}

This paper studies the role of hard information in contractual and market settings in which the receiver can flexibly adjust allocations and transfers in response to the sender's disclosure. These settings include monopoly pricing, bilateral trade with interdependent values, insurance contracting, and policy negotiations. Across these settings, the sender is worst off if she reveals her type completely: if her type becomes known, the receiver can adjust the terms of trade to extract her surplus. Taking this feature as our central departure from the literature, we characterize the entire set of equilibrium payoffs across these disclosure games. We establish an equivalence result: every payoff profile that can be achieved through information design can also be approximated by an equilibrium of the disclosure game. Thus, hard information enables the sender to attain her commitment payoff without having to commit to an information structure. Moreover, this result highlights how verifiability can empower the sender in settings in which bargaining power resides with the receiver.

\end{abstract}

\thispagestyle{empty} 
\vspace{0.4in}

\end{titlepage}

\setcounter{page}{1}

\section{Introduction}\label{Section-Introduction}

In many strategic interactions, one player can disclose hard information to influence others' actions. For example, the seller of an asset may reveal audited statements about an asset's value to attract higher price offers, or an insuree may share medical records to obtain better terms. An important literature initiated by \cite{grossman:81-JLE} and \cite{milgrom1981good} models such settings as \emph{disclosure} or \emph{persuasion games}. These games feature two core ingredients. First, the informed party---or the sender---possesses hard information that cannot be manipulated: she can share nothing but the truth although she need not share the entire truth. Second, she cannot precommit to what she discloses; of the utterances that she can make, she chooses the one that results in the most favorable outcome. The party hearing these utterances---the receiver---is no fool. He draws inferences from both what is said and what is left unsaid. Thus, what the sender says and how the receiver responds are jointly determined in {equilibrium}. 

The classical analysis of these settings shows that combining these ingredients leads to unraveling: the sender voluntarily discloses all information in the unique equilibrium. The logic for why a fully revealing equilibrium exists is that were the sender to withhold information, the receiver assumes the worst and responds in a way that makes concealment unprofitable. A more powerful force pushes every equilibrium to be fully revealing: in the settings modeled by this prior literature, for any pool of sender types, at least one type strictly prefers to reveal itself than remain pooled. The ability to disclose hard information thus creates a \emph{commitment trap}: even if the sender would prefer ex ante to conceal some information, strategic forces compel her to disclose everything.

While in some applications, the sender may prefer to reveal her type, this assumption fails in many economically important settings. Consider a buyer-seller interaction in which the buyer can disclose information about her value before the seller makes a take-it-or-leave-it (TIOLI) offer. The buyer gains nothing from revealing her true value because doing so enables the seller to extract her entire surplus. Similarly, in a monopolistic insurance market, an insuree may disclose some information to alleviate adverse selection but disclosing it all results in a contract that drives her utility down to her outside option. 

More generally, in contractual and market settings where the receiver can flexibly adjust the terms of trade based on what is disclosed and then issue a TIOLI offer, full disclosure leaves the sender no better off than her outside option. In such settings, she may prefer to withhold information to counteract the receiver's bargaining power but she cannot commit to it. Prior work has not generally understood the equilibrium implications of hard information in these contractual settings. Our goal is to provide a unified analysis. Does unraveling fail and, if so, what is the resulting set of equilibrium payoffs? 

We consider a broad class of sender-receiver games in which the sender discloses \emph{evidence} about her type to a receiver. Following \cite{grossman:81-JLE} and \cite{milgrom1981good}, the sender can disclose any subset of types that contains her true type. 
Upon receiving this disclosure, the receiver chooses an action. Our framework is agnostic about the nature of this action: it may be a price offer from a seller, a menu of insurance contracts from an insurer, or an allocation and transfer proposed by a principal. 

We make three assumptions. Our first assumption is that the sender weakly favors uncertainty: relative to her payoff from inducing any belief, she never strictly gains from fully revealing her type. This assumption is our primary departure from the existing literature, capturing the settings described above in which full disclosure results in the receiver extracting full surplus. 
The other two assumptions are standard. 
The second states that every message contains a ``worst-case type'' who no type capable of sending that message would like to imitate. The third assumption imposes a form of continuity on the receiver's payoff with respect to his actions and beliefs. These three assumptions are satisfied in many familiar settings, including monopoly pricing, asset sales, insurance contracting, and policy negotiations. 

Under our assumptions, unraveling fails: while full disclosure remains an equilibrium, there also exists an equilibrium in which the sender completely conceals her type. This multiplicity raises a central question: what is the range of equilibrium payoffs? In general, characterizing the entire equilibrium payoff set is difficult once unraveling fails; most existing approaches are therefore constructive and problem-specific. Nevertheless, our main result obtains this set, using information design as a tool. 

To that end, suppose the sender could commit to an information structure that reveals information to the receiver before the receiver chooses his action. A payoff profile $(u_S^*,u_R^*)$ is \emph{achievable} if there exists an information structure that induces these payoffs. Every equilibrium payoff of the disclosure game is achievable because any sender strategy can be replicated by an information structure. Our main finding, \Cref{Theorem-MainResult}, establishes the converse.  \smallskip
\begin{quote}
     \textbf{Main Result.} \emph{Every achievable payoff profile can be (approximately) supported by an equilibrium of the disclosure game.}
\end{quote} \smallskip
This result shows that, under our assumptions, there is virtually no gap between information design and voluntary disclosure: for every achievable payoff profile $(u_S^*,u_R^*)$ and every $\e>0$, some equilibrium of the disclosure game attains payoffs within $\e$ of $(u_S^*,u_R^*)$. Hence, equilibrium outcomes encompass not only full and no revelation but also the entire range of partially revealing outcomes. In particular, the sender could obtain payoffs arbitrarily close to that delivered by her optimal information structure, withholding information to counter the receiver's bargaining power. Thus, she can not only escape the commitment trap of unraveling but also attain her commitment payoff.

Methodologically, the result shows that in a broad class of economically important settings, one can characterize equilibrium payoffs by focusing solely on the receiver's obedience constraints---exactly as in information design---while disregarding the sender's incentive constraints altogether. Finally, we show that each of the three assumptions is essential: relaxing any of them allows one to find settings in which the conclusion of our main result fails. 

Our analysis also has implications for cheap-talk communication in these contractual settings. \Cref{Theorem-MainResult} implies that allowing for cheap talk---either alone or alongside hard information---does not enlarge the set of equilibrium payoffs relative to the disclosure game. In fact, in the applications discussed below, cheap talk by itself is completely ineffectual. 

We apply our result to market and contractual settings in which the receiver can adjust the terms of trade following disclosure; we also obtain additional results in these settings. We discuss these applications below.

\vspace{-.1in}
\paragraph{Monopoly Pricing:} We begin with the canonical monopoly pricing problem with unit demand. We show that our assumptions hold in this setting, leveraging a result by \cite{yang2023continuity}. Therefore, every payoff in the ``BBM triangle'' \citep{bergemann2015limits} can be virtually supported by an equilibrium of the disclosure game. We view this result as useful from several perspectives.

First, it responds to a challenge raised by \cite{bergemann2019information}, who caution against interpreting information design in markets literally because it \emph{``would need to identify an information designer who knew consumers’ valuations and committed to give partial information to the monopolist in order to maximize the sum of consumers’ welfare.''} Our finding shows that voluntary disclosure can serve as a microfoundation for this benchmark: information flows directly from the buyer to the seller in a standard disclosure game, without requiring a third party who observes the buyer's type or commits to an information structure.\footnote{\cite{fainmessergaleottimomot} and \cite{galpertiliuperego} develop important alternative microfoundations: the former studies a platform that uses past purchases to infer consumer valuations, and the latter models an intermediary who purchases consumer data and intermediates trade.} 

Second, the result highlights how an intermediary or data collective that verifies statements about consumers' values could enable consumers to attain payoffs close to their optimal information structure. In such equilibria, consumers disclose enough information to influence the monopolist to offer targeted discounts---grouping together high and low value buyers---but not so much information that the monopolist extracts their full surplus. This behavior resembles group-pricing schemes seen in practice, where senior citizens, students, and low-income households can disclose evidence to obtain price discounts. A key requirement to achieve the commitment benchmark is that the evidence be sufficiently rich, capable of pooling low and high values while excluding medium ones. In this regard, we prove that consumer payoffs are generally bounded away from the commitment benchmark if consumers could send only ``interval messages.'' Our analysis therefore speaks to the gains that consumers accrue when they  control, at a granular level, what data firms observe about them. 

While we focus on monopoly pricing, a similar logic extends to competitive markets with differentiated products. \cite*{elliottgaleottikohli} characterize the information structure that maximizes consumer surplus. We argue that equilibria of a multi-firm disclosure game can approach those payoffs. In these equilibria, consumers disclose selectively to firms, revealing their location to non-preferred firms so as to intensify competition, and partially disclosing information to the preferred firm to counteract its bargaining power.

Finally, we find that the truth-leaning refinement of \cite*{hart2017evidence} selects payoffs on the efficiency frontier, and virtually every efficient payoff can be supported by a truth-leaning equilibrium.

\vspace{-.1in}
\paragraph{Disclosures about Assets:} We study an asset market in which a seller strategically discloses evidence to mitigate adverse selection. Absent disclosure, prospective buyers fear purchasing a lemon, leading to inefficient trade. We find that, when facing a single buyer, the seller-optimal equilibrium entails partial disclosure: the seller reveals just enough to alleviate adverse selection while withholding information to retain all efficiency gains. In contrast, if there are two or more buyers, all equilibria feature full disclosure. These results offer a new perspective on how market competition promotes transparency.\footnote{The literature on competition and disclosure initiated by \cite{milgromroberts86} focuses on how multiple informed senders vie to influence a single uninformed receiver. By contrast, we study how multiple uninformed receivers affect the disclosure incentives of a single sender.} 

We use these findings to identify when the seller benefits most from attracting a second buyer. When adverse selection is severe, the seller gains less from a second buyer: partial disclosure then is a powerful tool that both alleviates adverse selection and counteracts the buyer's market power. When adverse selection is mild, however, competition dominates disclosure as a tool for limiting market power. Our analysis thus yields a joint prediction that links the severity of adverse selection to competition and transparency in the asset market.

\vspace{-.1in}
\paragraph{Insurance Contracting:} Should insurers be allowed to condition contracts on an insuree's disclosure of genetic information? Doing so improves efficiency by reducing the  informational asymmetry between the parties. Regulators, however, worry that allowing disclosure could compel insurees to reveal all information, leaving them worse off. Reflecting this concern, the \emph{Genetic Information Nondiscrimination Act} of 2008 prohibits US insurers from using genetic information provided by prospective insurees \citep{erwin08,golinghorst2022anti}. Against this backdrop, we evaluate how disclosure could shape the insurance market. We consider the standard insurance contract model and show that it satisfies our three assumptions. Thus, disclosure could enable the insuree to attain payoffs arbitrarily close to those generated by her (ex ante) optimal information structure. Yet, consistent with regulators' concerns, some equilibria can leave her strictly worse off than prohibiting disclosure. These findings highlight the potential role of government agencies or third-party intermediaries in  steering markets towards equilibria that benefit both insurers and insurees.

\vspace{-.1in}
\paragraph{Veto Bargaining:} In the Supplementary Appendix, we consider \cite{romer1978political}'s model of policy negotiations, a workhorse in political economy. We show that our assumptions and main finding hold in this setting despite the absence of transferable utility.

\vspace{-.1in}
\paragraph{Outline of Paper:} \Cref{Section-Model} describes the disclosure game and \Cref{Section-MainResult} the main results. \Cref{Section-Applications} considers  applications. \Cref{Section-Conclusion} concludes. Omitted proofs and material are in appendices. The rest of this introduction discusses the related literature.

\vspace{-.1in}
\paragraph{Related Literature:} 
We build on canonical disclosure models in which the sender discloses any set of types that contain her true type \citep{grossman1980disclosure,grossman:81-JLE,milgrom1981good,milgromroberts86}. In much of this work, the sender's payoff depends only on the receiver's posterior mean, motivated either through a market price or quantity choice. Such payoff structures limit the receiver in how he can vary contractual terms in response to the sender's disclosure. Our main departure is that we endow the receiver with the flexibility to choose prices and allocations, as in monopoly pricing or insurance contracting. This flexibility ensures that if the sender fully discloses her type, the receiver follows up with an action that extracts her surplus. This feature is why unraveling fails in our model.

A large literature explores other obstacles to unraveling, including \citeauthor{dye1985disclosure}-evidence structures and costly disclosure; see \cite{ben2025evidence} for a survey. The strand closer to our work studies settings that violate the payoff monotonicity conditions underlying the classical unraveling logic. \cite{okuno1990strategic}, \cite{seidmann1997strategic}, and \cite{mathis2008full} provide sufficient conditions for unraveling and they as well as \cite{giovannoni2007secrecy} and \cite{martini2018multidimensional} highlight how unraveling fails when those conditions are not met.\footnote{Also related is \cite{onuchicramos}'s result on how unraveling fails under group disclosure.} In a setting with state-independent preferences, \cite{titova2024persuasion} show that payoff monotonicity breaks when the receiver chooses from finitely many actions, and they identify conditions under which the sender-optimal achievable payoff can be supported by an equilibrium. Their analysis restricts the receiver's action space whereas, by contrast, we endow the receiver with greater flexibility.

Our work evaluates a form of non-monotonicity distinct from those above. Herein, the sender does not profit from revealing her type because the receiver can flexibly adjust the market price or contractual terms of trade based on her disclosure. On this theme, a few recent papers \citep{glode2018voluntary,pram2020,ali2023voluntary} model how disclosure can be beneficial in markets. A second feature that distinguishes our analysis from prior work is that we characterize the entire set of equilibrium payoffs, using information design as a tool. Because the prior literature primarily relies on constructive approaches, a complete characterizations is typically infeasible when unraveling fails.

In our analysis, we restrict ourselves to the evidence structure adopted by \cite{grossman:81-JLE} and \cite{milgrom1981good}, where the sender may disclose any set of types containing her true type. Our reasons are twofold. First, it isolates the role of the  payoff environment in generating our novel conclusions. Second, while richer evidence structures---such as stochastic evidence---could in principle enlarge the set of equilibrium payoffs, our main result implies that no evidence structure yields a larger payoff set in the settings that we study.

We take the standard Grossman-Milgrom evidence structure as a primitive. In other contexts, prior work models how a sender or intermediary would design evidence.\footnote{E.g., see \cite*{demarzo2019test}, \cite*{ali2022sell}, \cite{dasgupta2022hard}, \cite{pollrich2024irrelevance}, \cite{shishkin2024evidence}, and \cite{celik2025informative}.} \citet[pp. 2598-2599]{kamenica2011bayesian} show that voluntary disclosure can support the sender-optimal information structure in a game in which an uninformed sender publicly chooses a Blackwell experiment, privately observes its realization $s$, and then can send any message $m$ that contains $s$. More recently, \cite{arieli2025} and \cite{dai2025bayesian} study the scope and limits of selective disclosure of designed evidence.

Our work identifies settings in which the sender does not benefit from committing to an information structure.  A different literature studies mechanism design with evidence \citep*[e.g.,][]{glazer2004optimal,hart2017evidence,ben2019mechanisms} and identifies settings in which the receiver does not value committing to a mechanism.

\section{Model}\label{Section-Model}

We study a disclosure game in which a sender (she) shares evidence with a receiver (he) who then chooses an action. 
The sender privately observes her type $\theta$ drawn from the compact 
set $\Theta\subseteq \R^n$ according to the probability measure $F$ that admits a strictly positive density $f$ with respect to the Lebesgue measure on $\R^n$.
Her type determines what she can say: the sender of type $\theta$ chooses a message in $\mathcal{M}(\theta)\equiv \{m\in \mathcal C: \theta\in m\}$, where $\mathcal C$ denotes the collection of all non-empty closed subsets of $\Theta$. Like \cite{grossman:81-JLE} and \cite{milgrom1981good}, we interpret each message as \emph{evidence} in that any statement the sender makes, such as ``my type is in $m$," must be true.\footnote{Equivalently, we could have formulated evidence in the space of ``documents,'' as in \cite{lipman1995robust}, \cite{bull2004evidence}, \cite*{hart2017evidence}, and \cite*{ben2019mechanisms}.} Observe that $\mathcal{M}(\theta)$ includes the fully revealing message  $\{\theta\}$ (available only to type $\theta$), the fully concealing message $\Theta$ (available to every type), and a spectrum of messages that reveal some but not all information about the sender's type. 

After receiving the sender's message, the receiver chooses an action $a$ from a compact metrizable space $A$. The sender's payoff $u_S(a,\theta)$ is continuous in $a$ for each $\theta$, while the receiver's payoff $u_R(a,\theta)$ is upper semicontinuous in $a$ for each $\theta$. 

We now define strategies and equilibria for this game. The sender's strategy is a function $\rho:\Theta\rightarrow\Delta (\mathcal C)$ in which the support of $\rho(\theta)$ is a subset of $\mathcal M(\theta)$ for every type $\theta$. The receiver's strategy is a function $\tau:\mathcal C\rightarrow A$.\footnote{We endow $\mathcal{C}$ with the Hausdorff metric. Throughout our analysis, for a compact metrizable space $X$, $\Delta(X)$ denotes the set of (Borel) probability measures on $X$ endowed with the weak$^*$ topology. Our analysis allows for mixed actions; then, we interpret $A$ as the set of probability distributions over (pure) actions.} The receiver's beliefs about the sender's type are represented by the belief system $\mu: \mathcal{C} \to \Delta(\Theta)$. An assessment $((\rho, \tau), \mu)$ is a Perfect Bayesian Equilibrium (henceforth, equilibrium) if the following conditions hold: 
\begin{enumerate}[noitemsep,label=(\alph*)]
    \item \emph{Given her type, the sender discloses evidence optimally}: for every type $\theta$, $\rho(\theta)$ is supported on $\argmax_{m \in \mathcal{M(\theta)}} u_S(\tau(m),\theta)$, 
    \item \emph{Given the message, the receiver chooses actions optimally}: for every message $m$, $\tau(m)\in\argmax_{a \in A} \int_{\Theta} u_R(a,\theta) \, \mathrm{d}\mu(\theta | m)$,
    \item \emph{Beliefs respect evidence}: for every message $m$, the receiver's beliefs $\mu(m)$ have a support that is a subset of $m$.
    \item \emph{Bayes' Rule}: the beliefs $\mu$ are obtained from $F$ given $\rho$ using Bayes' rule, i.e., $\mu$ is a regular conditional probability system. 
\end{enumerate}

This framework accommodates numerous applications. Our primary objective is to capture market and contractual environments in which the receiver flexibly adjusts the terms of trade---allocations and transfers---based on information that is disclosed. We present two applications below, first considering the leading example of monopoly pricing and then a more general principal-agent setup. 
\begin{example}[Monopoly Pricing]\label{Example-BBM}
    Consider the standard monopoly pricing problem with incomplete information, augmented with a disclosure stage: the buyer of a good can disclose evidence about her value $\theta$ to a monopolist who then responds with a price offer that the buyer can accept or reject. Here, the buyer's type $\theta$ is drawn from $\Theta=[\underline\theta,\overline\theta]$ and action $a\in \Re$ represents a price. We write $u_S(a,\theta)=\max\{\theta-a,0\}$ for the buyer's payoff and $u_R(a,\theta)=a\mathbf{1}_{a\le \theta}$ for the monopolist's payoff, reflecting that if the buyer accepts the price $a$, trade happens at that price, and if she rejects, each party obtains $0$.\qed
\end{example}

\begin{example}[Principal-Agent Model]\label{Example-PrincipalAgent}
Consider a standard principal-agent setting in which the agent (sender) has a type $\theta$ drawn from $\Theta=[\underline{\theta},\overline{\theta}]$. The principal (receiver) offers a menu of pairs $(x,t)$, where $x\in X$ is an allocation from a compact space and $t\in\mathbb{R}$ is a transfer from the agent to the principal. The agent selects her preferred option from this menu or an outside option, $x_0\in X$. The agent's payoff from $(x,t)$ is $v_S(x,\theta)-t$ and the principal's payoff is $v_R(x,\theta)+t$. We assume that $v_S(\cdot,\theta)$ and $v_R(\cdot,\theta)$ are continuous for each $\theta$. 

Our setup captures this interaction by specifying that the principal's action $a$ is a compact subset of $X\times \mathbb{R}$ that includes the outside option, $(x_0,0)$. The agent selects her preferred option from $a$, breaking ties in the principal's favor. Letting $(x_a,t_a)$ denote type $\theta$'s optimal choice in $a$, the resulting payoffs are $u_S(a,\theta)=\max_{(x,t)\in a}\{v_S(x,\theta)-t\}$ and $u_R(a,\theta)= v_R(x_a,\theta)+t_a$.\qed
\end{example}

\section{When Disclosure Attains Design}\label{Section-MainResult}

\subsection{The Information Design Benchmark}\label{Section-Achievable}

We use the payoffs that stem from information design as a benchmark. In this benchmark, the receiver observes the realization of a Blackwell experiment and then chooses an action.  We call a distribution over types $G \in \Delta\left(\Theta\right)$ a \emph{belief}. A \emph{segmentation} is a distribution $\sigma \in \Delta\left(\Delta\left(\Theta\right)\right)$ over beliefs that average to the prior $F$---that is, $\int G \, \mathrm d \sigma(G) = F$---and a belief in the support of a segmentation is a \emph{segment}. A segmentation \emph{achieves} a payoff profile $(u_S^*,u_R^*)$ if player $i$'s ex ante expected payoff from the segmentation, given that the receiver best responds, is $u_i^*$; a payoff profile is \emph{achievable} if some segmentation achieves it. For instance, in monopoly pricing (\Cref{Example-BBM}), the set of achievable payoffs is the ``BBM-triangle'' characterized by \cite*{bergemann2015limits}, namely all feasible payoff profiles in which the monopolist does as well as she would from setting a uniform price.

Every equilibrium payoff profile of the disclosure game is achievable because every sender strategy induces a segmentation. Our main result pertains to the converse, namely conditions under which every achievable payoff can be approximated by an equilibrium of the disclosure game. We turn to these conditions next.

\subsection{The Key Assumptions}\label{Section-Assumptions}

We impose three assumptions. The first is that  the sender does not benefit from fully disclosing her type; this assumption is our primary departure from the literature and captures settings in which the receiver extracts the sender's surplus if she fully reveals her type. The second assumption, standard in the  literature, requires that every message contains a type that no other type prefers to imitate and can thus serve as the root for skeptical beliefs. The third assumption imposes a form of continuity on the receiver's utility function. 

To formalize these assumptions, let $a^*(G)\equiv \arg\max_{a\in A} \int u_R(a,\theta) \mathrm dG(\theta)$ be the receiver's optimal actions given belief $G$. 
Let $\underline{a}(\theta)$ denote a receiver-optimal action if the sender's type is known to be $\theta$, where the receiver breaks ties among his optimal actions so as to induce the lowest payoff for the type-$\theta$ sender.\footnote{In other words, $\underline{a}(\theta)\in\argmin\limits_{a\in a^*(\delta_{\theta})}u_S(a,\theta)$, where $\delta_\theta$ is the Dirac measure on $\theta$.} These definitions lead to our first assumption.

\begin{assumption}\label{Assumption-NoRents}
    The sender never strictly benefits from fully revealing her type relative to any other belief about her type. Formally, for every type $\theta$, belief $G$ such that the support of $G$ contains $\theta$, and receiver-optimal action $a\in a^*(G)$, $u_S(a,\theta)\ge u_S(\underline{a}(\theta),\theta)$.
    
\end{assumption}
\Cref{Assumption-NoRents} is our primary departure from the literature. A common assumption in prior work \citep[e.g.,][]{grossman:81-JLE,milgrom1981good} is that the sender's indirect utility over the receiver's beliefs is monotone: every type strictly prefers to induce ``higher'' beliefs in some order. This payoff monotonicity condition implies that in any pool of sender types, at least one type would prefer to separate and reveal itself. While natural in some environments, this condition rules out important contractual and market settings in which the receiver can flexibly adjust his action and transfers after observing the disclosure. In such settings, the sender cannot gain from full revelation as doing so would enable the receiver to make a TIOLI offer that fully extracts her surplus. In these settings, \Cref{Assumption-NoRents} holds instead. 

We illustrate this principle first in monopoly pricing (\cref{Example-BBM}): if the buyer reveals her type $\theta$, the monopolist would charge a price equal to $\theta$ that leaves her with no surplus. No other belief could make her worse off. 

More generally, consider the standard principal-agent setting described in \cref{Example-PrincipalAgent}. Under complete information, the principal offers a menu with two elements: the outside option with zero transfer and the efficient allocation paired with a transfer that extracts the sender's surplus. That is, the principal's action $a$ is the menu $\{(x_0,0),(x_\theta^*,t_\theta^*)\}$  where $x_0$ is the outside option, $x_\theta^*$ is an efficient allocation, maximizing $v_S(x,\theta)+v_R(x,\theta)$, and the transfer $t_\theta^*\equiv v_S(x^*_\theta,\theta)-v_S(x_0,\theta)$ leaves the sender just indifferent between accepting and rejecting.
As in monopoly pricing, complete information leads to a contract that yields a lower payoff than what the sender would obtain were she to induce any other belief.

We turn to our second assumption. 
\begin{assumption}\label{Assumption-WorstCase}
    Every message $m$ contains a \emph{worst-case type} $\hat\theta_m$ such that for every $\theta\in m$, 
    $u_S(\underline{a}(\theta),\theta)\geq u_S(\underline{a}(\hat{\theta}_m),\theta)$. In other words, no type in $m$ would prefer to imitate $\hat\theta_m$.
\end{assumption}
This assumption, standard in the literature, allows the modeler to deter off-path messages by endowing the receiver with ``skeptical beliefs'' that concentrate on a single type following an off-path message. It holds in settings that satisfy the payoff monotonicity condition \citep[e.g.,][]{grossman:81-JLE,milgrom1981good}: when every sender type prefers the receiver to hold higher beliefs, the worst-case type for a message is simply the lowest type capable of sending it. More generally, \cite{seidmann1997strategic} and \cite{hagenbach2014certifiable} show that a worst-case type exists if preferences satisfy a standard single-crossing condition. 

A similar logic applies in our leading applications, even though those applications violate the payoff monotonicity condition. For monopoly pricing, the worst-case type for a message $m$ is the highest type that could send it, $\hat\theta_m\equiv \max_{\theta\in m}\theta$. No other type that could send message $m$ would wish to mimic $\hat\theta_m$ because doing so induces the monopolist to charge a price equal to $\hat\theta_m$, yielding the buyer a payoff of $0$.

More generally, in the principal-agent setting of \cref{Example-PrincipalAgent}, \Cref{Assumption-WorstCase} holds whenever the agent's willingness to pay (WTP) for any alternative $x$, namely $v_S(x,\theta)-v_S(x_0,\theta)$, weakly increases in $\theta$. This increasing-WTP assumption holds in standard contracting environments: if $X$ is partially ordered with $x_0$ as the lowest alternative, the WTP is increasing if $v_S(x,\theta)$ satisfies increasing differences in $(x,\theta)$. In this setting, $\hat\theta_m\equiv \max_{\theta\in m}\theta$ serves as a worst-case type for message $m$. To see why, observe that given a point mass belief on $\hat\theta_m$, the principal optimally offers a menu of two items: the outside option and an efficient alternative for type $\hat\theta_m$ paired with the transfer that extracts full surplus from that type. Because type $\hat\theta_m$ has a higher WTP for this alternative than all lower types, every other type that could send message $m$ obtains zero surplus from this menu.

The discussion above emphasizes how \Cref{Assumption-NoRents,Assumption-WorstCase} are satisfied when the receiver makes a TIOLI offer. The example below shows that these assumptions are also compatible with the sender having some bargaining power. 

\begin{example}[Monopoly Pricing Revisited]
    Consider an adaptation of \Cref{Example-BBM} in which following disclosure, a \emph{random recognition} rule determines who makes the offer: the seller makes the offer with probability $\alpha\in (0,1)$ and the buyer makes the offer with probability $(1-\alpha)$. As the buyer's disclosure affects equilibrium payoffs only when the seller makes the offer, \Cref{Assumption-NoRents} still holds. Similarly, \Cref{Assumption-WorstCase} still holds, setting $\hat\theta_m\equiv \max_{\theta\in m}\theta$.\qed
\end{example}

Our final assumption imposes a form of continuity. To motivate this assumption, we observe that all of our applications, including monopoly pricing, are inherently discontinuous games in that $u_R(a,\theta)$ is discontinuous in $a$ for each $\theta$. Nevertheless, these games satisfy a form of continuity weaker than the standard notion. The condition below distills that underlying property. 
Let $U_R(G):=\max_{a\in A}\int u_R(a,\theta)\,\mathrm dG(\theta)$ be the receiver's expected payoff when he chooses his optimal action for a belief $G$. Of his best responses to belief $G$, let $\underline{a}(G)$ and $\overline{a}(G)$ be ones that minimize and maximize the sender's payoff: $\underline{a}(G)\in \arg\min_{a\in a^*(G)} \int u_S(a,\theta)\,\mathrm dG(\theta)$ and $\overline{a}(G)\in \arg\max_{a\in a^*(G)} \int u_S(a,\theta)\,\mathrm dG(\theta)$.

\newpage
\begin{assumption}\label{Assumption-Continuity}
The following hold:
    \begin{enumerate}[label=(\alph*)]
        \item The receiver's payoff from choosing an optimal action, $U_R(G)$, is continuous and the set of optimal actions, $a^*(G)$, is upper hemicontinuous. \label{a:continuity_1}
        \item For every belief $G$ and strictly positive $\e$ and $\delta$, there are
        \begin{itemize}
            \item a belief $\underline{H}$ such that the Radon-Nikodym derivative $\frac{\mathrm d\underline{H}}{\mathrm dG}\le 1+\e$, and any best response to $\underline{H}$ is in $B_{\delta}(\underline{a}(G))$; and
            \item a belief $\overline{H}$ such that the Radon-Nikodym derivative $\frac{\mathrm d\overline{H}}{\mathrm dG}\le 1+\e$, and any best response to $\overline{H}$ is in $B_{\delta}(\overline{a}(G))$.
        \end{itemize}
        Moreover, the functions that send $G$ to $\underline{H}$ and $\overline{H}$, respectively, are measurable. \label{a:continuity_2}
    \end{enumerate}
\end{assumption}

To elaborate on \ref{a:continuity_1}, let $u_R(a,G):=\int u_R(a,\theta)\,\mathrm dG(\theta)$ be the receiver's expected payoff when his belief is $G$. Were $u_R(a,G)$ continuous in $(a,G)$, part \ref{a:continuity_1} follows from Berge's maximum theorem. As noted above, $u_R(a,G)$ is \emph{not} continuous in $(a,G)$ in any of our applications. Nevertheless, we show that \ref{a:continuity_1} holds in these applications by applying maximum theorems that require weaker conditions than Berge's theorem. 

Part \ref{a:continuity_2}  specifies a form of ``selection continuity'' for the receiver's best responses: for every belief $G$, there exists a ``nearby'' belief $\underline H$ such that every best response to $\underline{H}$ lies close to the best response to $G$ that minimizes the sender's payoff. Symmetrically, there exists a nearby belief $\overline{H}$ such that every best response lies close to the best response to belief $G$ that maximizes the sender's payoff. A sufficient condition is that any ``target'' best response $a$ to belief $G$ can be ``uniquely selected'' by perturbing the belief $G$ to a ``nearby'' $H$ such that all best responses to $H$ are close to the target action $a$. This property holds if $A$ is convex and the receiver's utility function is strictly concave in $a$ for every belief $G$.

We show that \Cref{Assumption-Continuity}\ref{a:continuity_2} holds in our applications pertaining to monopoly pricing, asset markets, and policy negotiations using the fact that the receiver's action space is one-dimensional. An ancillary implication of part \ref{a:continuity_2} is that the receiver has a unique optimal action in the complete information game, which fails in our insurance application. To accommodate that application, we develop in \Cref{r:assn_insurance} a weaker form of continuity that does not require a unique optimal action and instead operates in the space of payoffs. The underlying economic idea is that the sender's risk aversion makes the receiver's payoff function ``almost'' strictly concave, which enables us to verify that this weaker notion holds.

\subsection{Main Result}

Given our assumptions, unraveling fails in this disclosure game: alongside a fully revealing equilibrium, there also exists a fully concealing equilibrium. In this equilibrium, the sender sends the message $\Theta$ regardless of her type and the receiver takes his (ex ante) optimal action on the equilibrium path. Following any off-path message, he assigns probability $1$ to a worst-case type for that message. Given \Cref{Assumption-NoRents,Assumption-WorstCase}, the sender finds neither full revelation nor a partially revealing off-path message profitable. 

Our main interest lies in characterizing the entire equilibrium set, not merely the extremes of full revelation and full concealment. Typically, an exhaustive characterization is infeasible once unraveling fails. In our setting, the full equilibrium payoff set can be identified through information design. 

\begin{theorem}\label{Theorem-MainResult}
    Suppose Assumptions \ref{Assumption-NoRents}, \ref{Assumption-WorstCase}, and \ref{Assumption-Continuity} are satisfied. For every achievable payoff profile $(u_S^*,u_R^*)$ and every $\e>0$, there is an equilibrium of the disclosure game whose payoffs are within $\e$ of $(u_S^*,u_R^*)$.
\end{theorem}

\Cref{Theorem-MainResult} shows that a rich set of equilibrium outcomes prevails in our setup. Importantly, the sender can obtain payoffs arbitrarily close to those delivered by her optimal information structure without having to commit to a Blackwell experiment. Voluntary disclosure thus enables the sender to obtain her commitment payoff rather than creating the commitment trap emphasized by \cite{grossman:81-JLE} and \cite{milgrom1981good} where she is forced to reveal everything.

To explain why we reach a different conclusion, we return to how the receiver herein can flexibly adjust the terms of trade following disclosure. Because the receiver can issue a TIOLI offer, the sender has no incentive to fully reveal her type. To the contrary, she has a strong interest in keeping the receiver in the dark so as to counteract his bargaining power. This force is not captured by existing models that either restrict the receiver to more limited instruments, such as a quantity choice with fixed prices \citep{milgrom1981good}, or take bargaining power away from the receiver altogether \citep{grossman:81-JLE}. Once the receiver has the ability to flexibly adjust the terms of trade, the sender strategically withholds information to counteract the receiver's bargaining advantage in the sender-optimal equilibrium.

We highlight two ancillary implications of \Cref{Theorem-MainResult}. First,  endowing the sender to also communicate by cheap talk---either alone or alongside hard information---does not enlarge the equilibrium set. In fact, in all applications that we consider, cheap talk alone effectively results in the babbling equilibrium and the receiver choosing his ex ante optimal action. Second, methodologically, \Cref{Theorem-MainResult} identifies a broad class of environments in which the equilibrium payoffs of a disclosure game can be obtained non-constructively solely through the receiver's obedience constraints, ignoring the sender's incentive constraints altogether. 

One may wonder if the equilibria that we study rely on unreasonable off-path beliefs. In the Supplementary Appendix,  \Cref{Theorem-IntuitiveCriterion} shows that all equilibria we consider satisfy the natural analogue of the intuitive criterion \citep{cho1987signaling}. 
Therein, we also study a second form of robustness, namely to the sender lacking evidence with small probability. This perturbation, in the spirit of \cite{dye1985disclosure},  motivates a large strand of the disclosure literature. We analyze its implications in the principal-agent setting of \Cref{Example-PrincipalAgent} and show that under an additional assumption---satisfied in monopoly pricing, asset markets, and insurance contracting---the equilibrium payoff set is robust to this perturbation. Thus, the possibility that the sender lacks evidence does not refine equilibrium outcomes in these contractual settings.
More specific to monopoly pricing, we show in \Cref{Section-Monopoly} that the truth-leaning refinement introduced by \cite{hart2017evidence} has bite, but only by ruling out inefficient payoff profiles.

\Cref{Section-Sketch} sketches the proof of \Cref{Theorem-MainResult}. The key obstacle is that in information design, a sender can commit to mixing across messages even if the resulting payoffs differ; by contrast, in an equilibrium of the disclosure game, the sender mixes only if she is indifferent. Our first step shows that if a segmentation is ``partitional''---in that only a measure-$0$ set of types belongs to more than one segment---then it can be supported in an equilibrium of the disclosure game. Our second step evaluates the cost of restricting attention to partitional segmentations: we show that for every achievable payoff profile, there exists a nearby payoff profile that is achieved by a finite partitional segmentation.\footnote{This payoff profile is nearby in an ex ante sense but may differ in its ex interim payoffs.} In the Supplementary Appendix, we show that our assumptions are necessary for \Cref{Theorem-MainResult} in that the conclusion fails if any assumption is dropped while maintaining the other two assumptions.

\subsection{Proof Sketch}\label{Section-Sketch}

The sender's equilibrium messaging strategy induces a segmentation: each message $m$ defines a segment corresponding to the receiver's belief following that message. Say that a segmentation $\sigma$ is \emph{finite} if its support is a finite set. A segmentation $\sigma$ is \emph{partitional} if for every $G, H \in \supp{(\sigma)}$ with $G\neq H$, $\supp(G) \cap \supp(H)$ has $F$-measure zero.

\begin{lemma} \label{Lemma-FinitePartitionalEqm}
    If \Cref{Assumption-NoRents,Assumption-WorstCase} are satisfied, then every payoff profile achieved by a finite partitional segmentation can be supported as an equilibrium. 
\end{lemma}

\noindent \textbf{Proof sketch.} Fix a finite partitional segmentation ${\sigma}$ and a best response $a:\supp{\sigma}\rightarrow A$ for the receiver. In the equilibrium we construct, the set of on-path messages is $\{\supp{G}\}_{G \in \supp{{\sigma}}}$. For any on-path message $\supp{G}$, the receiver updates her belief to $G$ and plays $a(G)$. For any off-path message $m$, the receiver updates to a point mass belief on the ``worst case type'' $\hat{\theta}_m$, which exists by \autoref{Assumption-WorstCase}, and then plays $\underline{a}(\hat{\theta}_m)$. 
For any type contained in the support of only one segment in $\supp{\sigma}$, there is only one on-path message available to her and the sender sends this message. For any type $\theta$ in the support of two or more segments in ${\sigma}$, the sender chooses the on-path message available to $\theta$ that results in the best possible action given $a(\cdot)$. 

We establish that the above constitutes an equilibrium. Observe that the receiver's beliefs are consistent with Bayes' rule.\footnote{Since ${\sigma}$ is finite partitional, the set of types that are contained in multiple segments have $F$-measure zero. Therefore, the behavior of such types does not affect the receiver's beliefs.} The receiver plays a best response after every message by construction. Also by construction, no type of the sender has a profitable deviation to an on-path message, and \autoref{Assumption-NoRents} and \autoref{Assumption-WorstCase} together imply that no type could profitably deviate to an off-path message.\hfill\qedsymbol

\begin{lemma}\label{Lemma-FinitePartitionalSuff}
    If \Cref{Assumption-Continuity} is satisfied, then for every achievable payoff profile $(u^*_S,u^*_R)$ and every $\e>0$, there is a finite partitional segmentation that achieves payoffs within $\e$ of $(u^*_S,u^*_R)$.
\end{lemma}

\noindent \textbf{Proof sketch.} Fix an achievable payoff profile $(u^*_S,u^*_R)$, a segmentation $\sigma$ and best response achieving this payoff profile, and $\e>0$. To illustrate the logic, we assume that the receiver plays $\overline{a}(G)$ for all $G \in \supp \sigma$. \autoref{fig:approx_process} depicts the three steps of the argument.

\begin{figure}[h!]
    \centering
    \begin{tikzpicture}[scale=1]
        \node at (0,0) {$\sigma$};
        \node at (4,0) {$\sigma_1$};
        \node at (8,0) {$\sigma_2$};
        \node at (12,0) {$\sigma_3$};
        \draw[->,thick] (0.5,0) -- (3.5,0);
        \draw[->,thick] (4.5,0) -- (7.5,0);
        \draw[->,thick] (8.5,0) -- (11.5,0);
        \node[below] at (2,-0.5) {\small close belief $\to$ close action};
        \node[below] at (6,-0.5) {\small finite};
        \node[below] at (10,-0.5) {\small finite partitional};
    \end{tikzpicture}
    \caption{The approximation process for \autoref{Lemma-FinitePartitionalSuff}.}
    \label{fig:approx_process}
\end{figure}

Because optimal actions may not be lower hemicontinuous in the belief, even small perturbations of the belief can significantly change the receiver's best response and thereby significantly alter the sender's payoff. We sidestep this issue by first perturbing any segment $G$ such that in the perturbed segment, all optimal actions are close to $\overline{a}(G)$: \autoref{Assumption-Continuity}\ref{a:continuity_2} assures that for every $G \in \supp \sigma$, there is a nearby segment $\overline{H}$ such that any best response to $\overline{H}$ is arbitrarily close to $\overline{a}(G)$. We obtain a new segmentation $\sigma_1$ by replacing each $G \in \supp \sigma$ with its corresponding $\overline{H}$.\footnote{To maintain Bayes' plausibility, we reduce the probability of each segment slightly and create one additional segment.} By choosing $\overline{H}$ sufficiently close to $G$, the receiver's payoff under $\sigma_1$ is within $\e/3$ of $u^*_R$ because by \autoref{Assumption-Continuity}\ref{a:continuity_1}, the receiver's payoff from choosing an optimal action is continuous in the belief. The sender's payoff under $\sigma_1$ is also within $\e/3$ of $u^*_S$ because any best response to $\overline{H}$ is arbitrarily close to $\overline{a}(G)$ and because the sender's payoff is continuous in the action for each type. 

The second step converts $\sigma_1$ to a finite segmentation $\sigma_2$ by ``merging'' segments in $\sigma_1$ that are sufficiently close to each other into a single segment, which is the average of the aforementioned segments. Because $\Delta(\Theta)$ is compact, the number of resulting ``average segments'' can be chosen to be finite. Because the receiver's best-response correspondence is upper hemicontinuous, any best response to an ``average segment'' in $\sigma_2$ is sufficiently close to any best response to any segment in $\sigma_2$ that is merged into it. Consequently, the payoffs under $\sigma_2$ are within $\e/3$ of those under $\sigma_1$. 

Our last step identifies a finite partitional segmentation $\sigma_3$ that achieves payoffs  within $\e/3$ of those under $\sigma_2$. Loosely speaking, we partition the type space into sufficiently small cubes, and approximate each of the finitely many segments in $\sigma_2$ using a collection of such cubes, which is possible because $F$ is absolutely continuous.
The statement on payoffs then follows from \autoref{Assumption-Continuity}\ref{a:continuity_1}.
Therefore, the payoffs under $\sigma_3$ are within $\e$ of $(u^*_S,u^*_R)$.\hfill\qedsymbol

\subsection{Finite Types}\label{Section-MainFinite}

The proof of \Cref{Theorem-MainResult} relies on the prior being atomless but a similar result holds for any finitely supported prior in which no single type has excessive mass. To establish this conclusion, we hold fixed the setting of \Cref{Section-Model}, including an ambient type space $\Theta\subseteq \Re^n$ that is compact, and consider a prior $F$ whose support is a finite subset of $\Theta$.\footnote{Observe that as to whether Assumptions \ref{Assumption-NoRents}, \ref{Assumption-WorstCase}, and \ref{Assumption-Continuity} hold is independent of the prior.}
\begin{theorem}\label{Theorem-Finite}
    Suppose Assumptions \ref{Assumption-NoRents}, \ref{Assumption-WorstCase}, and \ref{Assumption-Continuity} hold. For every $\e>0$, there is $\gamma>0$ such that if $F$ has finite support with $F(\{\theta\})\le \gamma$ for every type $\theta$, then for every achievable payoff profile $(u^*_S,u^*_R)$, there is an equilibrium of the disclosure game whose payoffs are within $\e$ of $(u^*_S,u^*_R)$.
\end{theorem}
\Cref{Theorem-Finite} offers a finite analogue of our main result. The key step shows that once types have sufficiently low mass, any finite segmentation can be approximated (in the appropriate sense) by a partitional one. The requirement that no type has excessive mass is necessary for our conclusion, as we illustrate below.

\begin{example}\label{Example-BinaryTypes}
    Consider an adaptation of \Cref{Example-BBM} in which the buyer's type $\theta$ is drawn from the binary set $\{\underline\theta,\overline\theta\}$ in which $\underline\theta>0$ and suppose that the optimal uniform price, $\overline{p}=\overline\theta$. The buyer-optimal segmentation would feature two segments, one comprising $\overline\theta$ alone and the other featuring a pool of both types such that the monopolist prices at $\underline\theta$. Although this setting satisfies \Cref{Assumption-NoRents,Assumption-WorstCase,Assumption-Continuity} (shown in \Cref{Section-Monopoly}), the payoff profile induced by this segmentation cannot be approximated by an equilibrium of the disclosure game. The issue is that, in equilibrium, type $\overline\theta$ mixes between messages $\{\overline\theta\}$ and $\{\underline\theta,\overline\theta\}$ only if the two messages result in the same price. Given that the former message results in a price of $\uv$, the buyer's equilibrium payoff must be $0$.\footnote{A different conclusion obtains were the buyer able to do interim information design \citep{perez2014interim,koessler2023informed,zapechelnyuk2023equivalence,madarasz2023information}. If the buyer publicly chooses a Blackwell experiment after learning her type, for every information structure, there exists an equilibrium in which all buyer types choose that information structure. \cite{clarkyang} propose a model of partially informed disclosure that nests both interim information design and standard disclosure games.}
\end{example}

\section{Applications}\label{Section-Applications}

We use our framework to understand how much information a consumer would disclose about her preferences in a market, when an asset seller discloses information to avoid an Akerlof lemon's problem, and how insurees would disclose information---e.g., genetic test results---in an insurance market. In the Supplementary Appendix, we also study policy negotiations with incomplete information. In these applications, we also develop additional findings.

\subsection{Monopoly Pricing}\label{Section-Monopoly}

Below, we show that the leading example of monopoly pricing satisfies the assumptions of our model. Our main result then implies that the set of achievable payoffs characterized by \cite{bergemann2015limits} can be supported virtually as equilibrium payoffs of the disclosure game. We sketch the logic for how equilibria of the disclosure game approximate the payoffs of the consumer-optimal segmentation and we characterize when such equilibria require non-interval messages. While our focus is on monopoly pricing, we show that this logic also applies in competitive markets with differentiated products. Finally, we show that the truth-leaning refinement selects equilibrium payoffs on the efficiency frontier. 

The model corresponds to that of \Cref{Example-BBM}: a monopolist (he) sells a product to a single consumer (she), who demands a single unit. The consumer's valuation $\theta$ is drawn according to an absolutely continuous CDF $F$ with support on $\Theta = [\uv,\ov]$ where $\uv \geq 0$. The monopolist's reservation value is $0$. We augment this standard game with a disclosure stage according to the following timeline. The consumer first observes $\theta$ and sends a message $m \in \mathcal{M}(\theta)$ to the monopolist. The monopolist then sets a price $a \in [0,\overline\theta]$. A type-$\theta$ consumer's payoff is $u_S(a,\theta) = \max\{\theta - a, 0\}$, and the monopolist's payoff is $u_R(a,\theta) = a \mathbf{1}_{a \le \theta}$. 

\cite{bergemann2015limits} show that a payoff profile is achievable so long as (i) the consumer's payoff is nonnegative, (ii) the monopolist's payoff is no less than $\overline{u}_R$, his payoff from charging the optimal uniform price, and (iii) the total payoff is no more than the maximal aggregate surplus $\overline w$. We have already shown that this setting satisfies \Cref{Assumption-NoRents,Assumption-WorstCase}. We establish that \Cref{Assumption-Continuity} also holds; this argument is involved and for one of the steps, we use \cite{yang2023continuity}. Therefore, every achievable payoff can also be approximated through an equilibrium of the disclosure game. 

\begin{proposition}\label{Proposition-BBM}
    For every payoff profile $(u_S^*, u_R^*)$ with $u_S^*\ge 0$, $u_R^*\ge \overline u_R$, and $u_S^*+u_R^*\le \overline w$, and every $\e>0$, there is an equilibrium of the disclosure game that supports payoffs within $\e$ of $(u_S^*,u_R^*)$.
\end{proposition}

Conceptually, this result highlights how voluntary disclosure offers a potential microfoundation for information design in markets: all payoff profiles in the BBM triangle can be (virtually) supported with hard information that flows directly from the consumer to the seller, without requiring an intermediary to know the consumer's value. As highlighted earlier, \cite{bergemann2019information} caution against interpreting information design in markets literally as it would require an intermediary that both knows the consumer's value and can commit to an information structure. Our result shifts the focus from such an intermediary to one who can simply \emph{verify} claims about consumer valuations.\footnote{\cite{fainmessergaleottimomot} and \cite{galpertiliuperego} offer alternative foundations for how an intermediary may learn the consumer's value: the former studies a platform that learns from past purchases and the latter analyzes an intermediary that purchases consumer data.} 

Given the importance of this setting, we offer a heuristic sketch of how we approximate  payoffs of the consumer-optimal segmentation in an equilibrium of the disclosure game. Following that sketch, we establish additional results.

\begin{example}\label{Example-Pareto}
For expositional simplicity, we consider a prior $F$ that is the equal mixture of two Pareto distributions, one starting at $1$ and the other at $2$. A Pareto distribution creates an ``equirevenue'' demand curve in that that all prices above the starting point yield the same revenue to the monopolist. The consumer-optimal segmentation splits the aggregate market into two segments: a low-price market that is the Pareto distribution starting at $1$, denoted $F_1$ in \Cref{Figure-Pareto} (unnormalized), and a high-price market that is the Pareto distribution starting at $2$, denoted $F_2$. The monopolist charges prices $1$ and $2$ respectively in the two segments.

\begin{figure}[h!]\centering
		        \begin{tikzpicture}[scale=3.5]
        {\draw [smooth,samples=100,domain=.2:.6,thick] plot({\x},{1/2-1/2*.2/\x});
        \draw [smooth,samples=100,domain=.6:2,thick] plot({\x},{1-.1/\x-.3/\x}) node[right] {\small $F$};}
            \draw[thick, <->] (0,1.15) -- (0,0) -- (2.15,0) node[below]{\small $\theta$};
        {
        \node[below] at (0.2,0) {1};
        \node[below] at (0.6,0) {2};
        }
        {
        \draw [smooth,samples=100,domain=.2:2,cred,thick] plot({\x},{1/2-1/2*.2/\x}) node[right] {\small $F_1$};
        \draw [smooth,samples=100,domain=.6:2,webblue,thick] plot({\x},{1/2-.3/\x}) node[right] {\small $F_2$}; 
        \node[below, cred] at (0.2,0) {1};
        \node[below, webblue] at (0.6,0) {2};
        }
        \end{tikzpicture}
        \caption{\small We depict the prior CDF $F$. The consumer-optimal segmentation splits the market into $F_1$ (unnormalized), which elicits a price of $1$ from the monopolist, and $F_2$, which elicits a price of $2$.}\label{Figure-Pareto}
	\end{figure}

 We first show that no equilibrium of the disclosure game can attain the payoffs from this segmentation exactly. The segmentation splits types $\theta \geq 2$ to two different segments, so replicating it in an equilibrium of the disclosure game would require these types to randomize across messages, some leading to a price of $1$ and  others to a price of $2$. But a buyer would mix only if the two messages yielded her the same payoff. Since the buyer strictly prefers the message that leads to the lower price, such randomization is untenable without commitment.

 However, we can approximate these payoffs. In \Cref{Figure-Approximation}, we illustrate a segmentation in which (i) all types in $m_1$ send message $m_1$, (ii) all types in $m_2$ (except those on the boundary with $m_1$) send message $m_2$, and (iii) all sufficiently high types send message $m_3$. The messages $m_1$ and $m_2$ are constructed so that the receiver's posterior beliefs closely approximate $F_1$ and $F_2$, respectively, up until the point where types enter $m_3$. We show that this segmentation can be supported by an equilibrium of the disclosure game and that its payoffs approximate those of the consumer-optimal segmentation.

  \begin{figure}[t!]\centering
    \begin{tikzpicture}[scale=3.5]
        {
        \draw [smooth,samples=100,domain=.2:.6,thick,dotted] plot({\x},{1/2-1/2*.2/\x});
        \draw [smooth,samples=100,domain=.6:2,thick,dotted] plot({\x},{1-.1/\x-.3/\x}) node[right] {\small $F$};
        \draw [smooth,samples=100,domain=.2:2,cred,dotted,thick] plot({\x},{1/2-1/2*.2/\x}) ;
        \draw [smooth,samples=100,domain=.6:2,webblue,dotted,thick] plot({\x},{1/2-.3/\x}) ;}
        {
        \draw [smooth,samples=100,domain=.2:.6,cred,thick] plot({\x},{1/2-1/2*.2/\x});
        \draw[cred,thick] (.2,-.1) -- (.6,-.1)  ;
        \draw[cred] (.4,-.2) node {$m_1$};
        \draw[fill,cred] (.6,-.1) circle;}
        {
        \draw [smooth,samples=100,domain=.6:.8,webblue,thick] plot({\x},{1-.1/\x-.3/\x-(1-.1/.6-.3/.6)});
        \draw[cred,thick] (.6,1/2-1/2*.2/.6) -- (.8,1/2-1/2*.2/.6);
        \draw[webblue,thick] (.6,-.1) -- (.8,-.1); 
          \draw[webblue] (.7,-.2) node {$m_2$};
          \draw [smooth,samples=100,domain=.6:.8,thick] plot({\x},{1-.1/\x-.3/\x}); node[right];
        \draw[fill,webblue] (.8,-.1) circle;}
        {
        \draw [smooth,samples=100,domain=.8:1,cred,thick] plot({\x},{1-.1/\x-.3/\x-(1-.1/.8-.3/.8)+(1/2-1/2*.2/.6)});
        \draw[cred] (.9,-.2) node {$m_1$};
        \draw[webblue,thick] (.8,{1-.1/.8-.3/.8-(1-.1/.6-.3/.6)}) -- (1,{1-.1/.8-.3/.8-(1-.1/.6-.3/.6)}); 
         \draw [smooth,samples=100,domain=.8:1,thick] plot({\x},{1-.1/\x-.3/\x}); 
        \draw[cred,thick] (.8,-.1) -- (1,-.1);} 
        {
        \draw [smooth,samples=100,domain=1:1.7,webblue,thick] plot({\x},{1-.1/\x-.3/\x-(1-.1/1-.3/1)+(1-.1/.8-.3/.8-(1-.1/.6-.3/.6))});
        \draw[cred,thick] (1,{1-.1/1-.3/1-(1-.1/.8-.3/.8)+(1/2-1/2*.2/.6)}) -- (1.7,{1-.1/1-.3/1-(1-.1/.8-.3/.8)+(1/2-1/2*.2/.6)});
        \draw [smooth,samples=100,domain=1:1.7,thick] plot({\x},{1-.1/\x-.3/\x}); 
        \draw[webblue] (1.35,-.2) node {$m_2$};
        \draw[webblue,thick] (1,-.1) -- (1.7,-.1);} 
        {
        \draw [smooth,samples=100,domain=1.7:2,brown,thick] plot({\x},{1-.1/\x-.3/\x-(1-.1/1.7-.3/1.7)});
        \draw[brown,thick] (1.7,-.1) -- (2,-.1); 
        \draw[brown] (1.85,-.2) node {$m_3$};
        \draw[cred,thick] (1.7,{1-.1/1-.3/1-(1-.1/.8-.3/.8)+(1/2-1/2*.2/.6)}) -- (2,{1-.1/1-.3/1-(1-.1/.8-.3/.8)+(1/2-1/2*.2/.6)}); 
         \draw [smooth,samples=100,domain=1.7:2,thick] plot({\x},{1-.1/\x-.3/\x}); 
        \draw[webblue,thick] (1.7,{1-.1/1.7-.3/1.7-(1-.1/1-.3/1)+(1-.1/.8-.3/.8-(1-.1/.6-.3/.6))}) -- (2,{1-.1/1.7-.3/1.7-(1-.1/1-.3/1)+(1-.1/.8-.3/.8-(1-.1/.6-.3/.6))}); 
        };
        \draw[thick, <->] (0,1.15) -- (0,0) -- (2.15,0) node[below]{\small $\theta$};
        \end{tikzpicture}
        \caption{\small The above constructs three market segmentations, one of which approximates $F_1$, one of which approximates $F_2$, and the third contains some high types. }\label{Figure-Approximation}
        \end{figure}

 To address the first point, note that there are only three on-path messages. We rule out deviations to any off-path message $m$ by assigning skeptical beliefs that put probability $1$ on the worst-case type $\max_{\theta\in m} \theta$. Given such beliefs, the monopolist charges a price that equals that type. These resulting prices make deviating to off-path messages unprofitable for the consumer, including those messages that fully reveal her type. We then consider deviations to other on-path messages. The only types that access more than one on-path messages are those at the (pairwise) intersections of $m_1$, $m_2$, and $m_3$. Our construction assures that all boundary types send the message that results in the lower price. 

 Finally, we argue that this segmentation delivers payoffs close to those of the consumer-optimal segmentation. Let $\tilde\theta$ denote the lowest type in $m_3$. Observe that given beliefs $F_1$, the monopolist is indifferent between charging a price of $1$ and $1+\e$, which reflects that his inframarginal gains from raising the price beyond $1$ is exactly offset by his marginal loss from excluding consumers in $[1,1+\e)$. Truncating the belief to $[1,\tilde\theta]$ breaks his indifference: the truncation reduces the gains from charging a higher price but leaves the loss undiminished, making $1$ the \emph{uniquely} optimal price. A message $m_1$ that induces beliefs close to the truncated $F_1$ must then elicit a price close to $1$: the appropriate analogue of Berge's Theorem implies that beliefs converging to the truncated $F_1$ must induce prices that converge to the uniquely optimal limit price. By the same reasoning, the monopolist's price following message $m_2$ must be close to $2$. Given that the mass of types in $m_3$ can be made arbitrarily small, ex ante payoffs approximate those of the consumer optimal segmentation. \qed
\end{example}

The construction above sketches how an equilibrium of the disclosure game can deliver payoffs close to those of the consumer-optimal segmentation. One may wonder what structural features  consumer-optimal equilibria must exhibit. We show below that a general property of these equilibria is that some low-price segments---such as segment $m_1$ in the construction---must pool together high and low types while excluding intermediate ones. This non-monotone pooling is not an artifact of the example but is instead intrinsic to attaining the consumer-optimal benchmark.

To formalize this finding, we say that an equilibrium uses \emph{interval messages} if for almost every type $\theta$, its equilibrium message $\rho(\theta)$ is a closed interval. We also call an aggregate market \emph{unimprovable} if the monopolist's optimal uniform price---i.e., the price he would charge based solely on his prior---equals $\uv$, the lowest type. Unimprovable markets are atypical; in most settings, the monopolist finds it optimal to charge prices that exclude consumers with low valuations. The result below links these two notions.
\begin{proposition}\label{Proposition-Interval}
    There are equilibria in interval messages that approximate the consumer-optimal segmentation if and only if the aggregate market is unimprovable.
\end{proposition}
\Cref{Proposition-Interval} shows that interval messages suffice only when the monopolist would serve all consumer types in the aggregate market; otherwise, any consumer-optimal equilibrium must feature segments that pool some types while excluding those in between. The ``if'' direction is immediate: if the aggregate market is unimprovable, the fully concealing equilibrium---which uses only interval messages---attains the payoffs of the consumer-optimal segmentation. The ``only if'' direction is more subtle. The key idea is that if an interval message pools low and high types to steer the monopolist to offer a price discount, all intermediate types have the option to mimic those types and secure the same discount. The monopolist is willing to serve all of these types only if the aggregate market is unimprovable.

We interpret these results in the context of regulatory measures that give consumers  control of their data. Our results suggest that an intermediary or data collective capable of verifying statements about consumer valuation could substantially benefit consumers all the way to the payoffs of the consumer-optimal segmentation. Crucial to this prospect is that consumers can verifiably disclose rich statements, certifying that they belong to a ``bucket'' without having to disclose their exact valuation. This result connects to group-pricing schemes observed in practice in which a third-party intermediary verifies that a consumer belongs to a designated group---such as senior citizens, students, or low-income households---and provides evidence that triggers a price discount. Our results also highlight the power of finetuned disclosures: while interval-style messages may benefit consumers, consumers may benefit from using more elaborate messages that pool together extreme valuations. 

Hard or verifiable information is crucial for the consumer to obtain such gains: were she able to use only cheap-talk messages, all equilibria essentially collapse to the babbling outcome in which the monopolist ignores their messages and charges the optimal uniform price.\footnote{\cite{hidir2021privacy} show that cheap talk can be effective when the monopolist chooses both price and product design.} 

While our analysis focuses on monopolistic markets, a similar logic applies to  competitive markets with differentiated products. Suppose there are two or more firms, and a consumer's type $\theta\equiv (\theta_1,\ldots,\theta_n)$ encodes her valuation for each of the $n$ products. Consider the disclosure game in which the consumer of type $\theta$ privately discloses a message $m$ in $\mathcal M(\theta)$ to each firm. We argue that there exist equilibria of this multi-firm disclosure game whose payoffs approximate those of the consumer-optimal segmentation characterized by \cite*{elliottgaleottikohli}. In these equilibria, the consumer fully reveals her type to all firms except her favorite one, which induces all her non-preferred firms to compete aggressively for her business, setting a price that equals marginal cost. To her favorite firm, she discloses information only partially, mirroring the disclosure strategy in the monopoly problem. The favored firm then prices optimally given the disclosure it receives, knowing that all other firms are competing intensely for the consumer's business.

Returning to monopoly pricing, \Cref{Proposition-BBM} shows that some equilibria of the disclosure game are inefficient. Our next result proves that a refinement standard in disclosure games selects efficient equilibria. This refinement is the ``truth-leaning equilibrium''  proposed by \cite*{hart2017evidence} who---invoking the Twainian adage, ``When in doubt, tell the truth''---study limit equilibria of perturbed games in which the sender accrues an infinitesimal gain if she shares the whole truth.  

Formally, for a function $\boldsymbol{\e}: \Theta \to \Re_{>0}$, consider the perturbed game $\Gamma^{\boldsymbol{\e}}$ in which the consumer's payoff increases by $\boldsymbol{\e}(\theta)$ when the type is $\theta$ and she sends message $\{\theta\}$. An equilibrium $((\rho,\tau),\mu)$ of the original game is \emph{truth-leaning} if there exist (i) a sequence of functions $\boldsymbol{\e}^n$ that converges uniformly to $\mathbf{0}$, where $\mathbf{0}$ is a constant function that maps every $\theta$ to 0, and (ii) a sequence $((\rho^n,\tau^n),\mu^n)$ that converges uniformly to $((\rho,\tau),\mu)$ such that for each $n \in \mathbb{N}$, $((\rho^n,\tau^n),\mu^n)$ is an equilibrium of the perturbed game $\Gamma^{\boldsymbol{\e}^n}$.\footnote{\cite*{hart2017evidence} also require that in any perturbed game, every type of the sender fully reveal her type with positive probability. This requirement has no bite in our setting.} \autoref{Proposition-TruthLeaning} characterizes truth-leaning equilibria of this game.

\begin{proposition} \label{Proposition-TruthLeaning}
    The following hold:
    \begin{enumerate}[label=(\alph*)]
        \item \label{p:truthleaningpart1} The payoff profile of every truth-leaning equilibrium is efficient:  $(u_S^*,u_R^*)$ is supported by a truth-leaning equilibrium only if $u_S^*+u_R^*=\overline w$.
        \item \label{p:truthleaningpart2} For every efficient payoff profile, there is a nearby payoff profile supported by a truth-leaning equilibrium: for every $(u_S^*,u_R^*)$ with $u_S^*\geq 0$, $u_R^*\geq \overline u_R$, and $u_S^*+u_R^*=\overline{w}$, and every $\e>0$, there is a truth-leaning equilibrium of the disclosure game that supports payoffs within $\e$ of $(u_S^*,u_R^*)$.\footnote{Equivalently, \Cref{Proposition-TruthLeaning}(b) identifies that there is a dense set of payoff profiles on the efficiency frontier of the BBM triangle that can be supported by (truth-leaning) equilibria of the disclosure game.}
    \end{enumerate}
\end{proposition}

Here is the logic. The only scope for inefficiency is that in some segment, the optimal price exceeds some consumer types. 
Such behavior cannot arise in an equilibrium of the perturbed game because those consumer types would then be better off revealing the whole truth to accrue the infinitesimal bump. 
Consequently, every truth-leaning equilibrium is efficient. 
A more subtle intuition underlies why all efficient payoff profiles can be approximately supported by a truth-leaning equilibrium. We show that for every finite partitional equilibrium that supports payoff $(\tilde u_S,\tilde u_R)$, there is an efficient truth-leaning equilibrium that supports payoffs $(\tilde u_S,\overline{w}-\tilde u_S)$. \Cref{Proposition-BBM} then implies that any efficient payoff profile can be approached using a truth-leaning equilibrium.

\subsection{Disclosures about Asset Quality}\label{Section-Asset}

We study a financial market in which an asset seller chooses how much information to divulge prior to sale. We show that when facing a single buyer, the seller strategically discloses information partially, revealing just enough to both alleviate adverse selection and diminish the buyer's bargaining power. By contrast, if multiple buyers compete for the asset, the unique equilibrium involves the seller releasing all information. Competition thus spurs disclosure. We then use our characterization of the seller-optimal equilibrium to identify when the seller has strongest incentives to recruit a second buyer. 

We begin with the setting of a single buyer. A seller (she) sells a single unit of an asset. Its quality is $\theta$ drawn according to an absolutely continuous CDF $F$ with support on $\Theta = [\uv,\ov]$. The seller's value of keeping the asset is $c(\theta)\geq 0$, which is weakly increasing and continuous in $\theta$. The buyer's value for the asset is $v(\theta)$, which strictly exceeds $c(\theta)$, is strictly increasing and continuous in $\theta$. Prior to sale, the seller can disclose evidence: given an asset quality $\theta$, the seller sends a message $m \in \mathcal{M}(\theta)$ to the buyer. The buyer then sets a price $a \in [0,v(\ov)]$, and the seller chooses to accept or reject the offer. The seller's payoff is $u_S(a,\theta) = \max\{a, c(\theta)\}$, and the buyer's is $u_R(a,\theta) = (v(\theta)-a) \mathbf{1}_{a \ge c(\theta)}$. Observe that $u_S(\cdot, \theta)$ is continuous and $u_R(\cdot, \theta)$ is upper semicontinuous for every type $\theta$.

Total surplus is maximized by having trade, resulting in an ex ante surplus of $\E[v(\theta)]$. However, absent disclosure, this market suffers from adverse selection \citep{akerlof1970market} given that at a given price $a$, only types $\theta$ for whom $c(\theta)\leq a$ would be willing to trade. The buyer's payoff then is $\overline{u}_R\equiv\max_{a\in [0,v(\ov)]}\int_{\{\theta\in \Theta:c(\theta)\leq a\}} (v(\theta)-a) \,\mathrm{d}F(\theta)$.

The set of achievable payoffs of this setting, characterized by \cite{kartik2024lemonade}, corresponds to all payoff profiles in which total surplus is no more than $\E[v(\theta)]$, the seller's payoff is no less than $\E[c(\theta)]$, and the buyer's payoff is no less than $\overline{u}_R$. Because our three assumptions hold in this disclosure game---the sender's complete information payoff $c(\theta)$ is lower than any incomplete information payoff, every message $m$ has a worst-case type $\min_{\theta\in m} \theta$, and we establish continuity using an approach from variational calculus---every achievable payoff can be virtually supported as an equilibrium. 
\begin{proposition}\label{Proposition-Akerlof}
    Suppose that there is a single buyer. Then for every payoff profile $(u_S^*, u_R^*)$ with $u_S^*\ge \E[c(\theta)]$, $u_R^*\ge \overline u_R$, and $u_S^*+u_R^*\le  \E[v(\theta)]$, and every $\e>0$, there is an equilibrium of the disclosure game that supports payoffs within $\e$ of $(u_S^*,u_R^*)$.
\end{proposition}
An implication of \Cref{Proposition-Akerlof} is that, given a single buyer, the seller can partially disclose evidence so that trade is nearly fully efficient and yet the buyer does not obtain more than his ``Akerlof payoff'' $\overline{u}_R$. Thus, the seller can disclose evidence that counteracts the lemon's problem while reaping all those gains herself. 

Now suppose that there are two or more buyers, each of whom has value $v(\theta)$ for an asset of quality $\theta$. We consider a disclosure game in which the seller sends a message $m$ in $\mathcal M(\theta)$ publicly to the buyers, each buyer then simultaneously offers a price in $[0,v(\ov)]$, and the seller then accepts one offer if any. In equilibrium, each buyer offers a price equal to his expected value of the asset conditional on all available information, as per standard Bertrand logic. We argue that this competitive price-setting leads to full unraveling. 
\begin{proposition}\label{Proposition-Auctions}
	Suppose that there are two or more buyers. Then every equilibrium is outcome equivalent to the fully revealing equilibrium in which the seller discloses $\{\theta\}$ for every asset quality $\theta$.
\end{proposition}
To see why this result holds, suppose towards a contradiction that there were an equilibrium with partial disclosure. Consider a pool of types that all send some message $m$, inducing belief $G$, and let $\hat\theta$ be the highest type in the support of $G$. In this putative equilibrium, the seller of type $\hat\theta$ cannot obtain a payoff higher than $\max\{\E_G[v(\theta)],c(\hat\theta)\}$: either she sells the asset at a price bounded above by $\E_G[v(\theta)]$ or she keeps the asset. Disclosing the message $\{\hat\theta\}$ is a profitable deviation because it fetches a strictly higher price, $v(\hat\theta)$. Hence, for every pool, the highest type profits from separating thereby unraveling any pooling equilibrium. Central to this logic is that the competitive market results in a disclosure game that satisfies \citeauthor{grossman:81-JLE}-\citeauthor{milgrom1981good}'s payoff monotonicity condition: as $v$ is strictly increasing, the seller would prefer to induce belief $G'$ over $G$ whenever $G'$ first-order stochastically dominates $G$.

The contrast between \Cref{Proposition-Akerlof,Proposition-Auctions} offers a new perspective on how competition promotes disclosure. \cite{milgromroberts86} initiated the study of competition and disclosure, focusing on how competition among multiple informed senders induces greater disclosure to a single uninformed receiver. Our results highlight the opposite force: here, competition among multiple uninformed receivers compels a single sender to divulge all information. Their competition ensures that the sender reaps all the gains from trade, leaving her with no profitable reason to withhold information.

We use this result to evaluate when a seller has the strongest motive to recruit buyers. In the competitive market, her ex ante payoff is $\E[v(\theta)]$ whereas when facing a single buyer, she accrues approximately $\E[v(\theta)]-\overline{u}_R$ in the seller-optimal equilibrium. Recall that $\overline{u}_R$ is the buyer's ``Akerlof payoff,'' corresponding to what he obtains in the lemon's problem without disclosure. All else equal, the seller then has less to gain from a second buyer when the Akerlof payoff is lower. If adverse selection is severe ($\overline u_R \approx 0$), she can kill two birds---adverse selection and the buyer's bargaining power---with the single stone of partial disclosure.

\subsection{Insurance Contracting}\label{Section-Insurance}

Should an insurance company be allowed to condition its contracts based on an insuree's disclosures? This question has been salient recently in the context of genetic testing, with regulations that prohibit insurers from using genetic information. We evaluate this debate within the context of our framework. We show that the standard monopolistic insurance model \citep[e.g.,][]{stiglitz1977monopoly,chade2012optimal} satisfies our three main assumptions. Our main finding then implies that an insuree could use her hard information to counteract the bargaining power of the insurer and attain the payoffs of her optimal information structure.\footnote{Our work complements recent work on genetic testing in insurance markets: \cite{pram2023learning} focuses on costs of acquiring evidence and \cite{azevedo2025} study how genetic prediction influences the degree of adverse selection.} 

An insuree has initial wealth $w>0$, faces a potential loss $\ell\in(0,w)$ with probability $\theta$, and has risk preferences represented by a strictly increasing, continuously differentiable, and strictly concave (Bernoulli) utility function $v:\R_{\ge 0}\rightarrow \R$. The probability of loss, $\theta$, is the insuree's \emph{type}, and is her private information. The insuree's outside option at every stage is to purchase zero insurance.

The insurer is risk neutral and has beliefs about the insuree's type given by the absolutely continuous CDF $F$ with density $f$ and support $\Theta := [\uv,\ov] \subseteq (0,1)$. Without loss, the insurer chooses a menu of contracts $(x(\theta),t(\theta))\in \R^2$ for each type $\theta$ comprising a premium $t(\theta)$ and an indemnity payment $x(\theta)$ in the event of a loss, subject to incentive and participation constraints. The expected profit from a contract $(x,t)$ chosen by a type-$\theta$ insuree is $t-\theta x$.

We append a disclosure stage to this problem. After observing her type $\theta$, the insuree sends a message $m \in \mathcal{M}(\theta)$ to the insurer. The insurer then offers a menu of contracts and the insuree selects one of the contracts or chooses no insurance.

Formulated this way, it is difficult to verify our continuity notion directly. Therefore, we reformulate this game such that the insurer directly chooses the expected utility of every type of the insuree instead of offering a menu of contracts \citep[see][]{chade2012optimal}. 
Observe that any incentive-compatible menu of contracts $(x(\theta),t(\theta))$ can be reformulated as a menu 
$(D(\theta), a(\theta))$, where
\begin{align*}
    D(\theta)&:= v(w-t(\theta)) - v(w-\ell + x(\theta)-t(\theta)),\\
    a(\theta)&:= v(w-t(\theta))-\theta D(\theta).    
\end{align*}
If a type-$\theta$ insuree accepts the menu, $v(w-t(\theta))$ is her utility when no loss occurs, $D(\theta)$ is the drop in utility when she suffers a loss, and $a(\theta)$ is her \emph{indirect utility}. 

By standard arguments, incentive compatibility implies that the indirect utility function $a(\cdot)$ is convex; in this case, $a'(\theta)=-D(\theta)$ almost everywhere. 
Given the insurer's belief $G$, any optimal menu satisfies the properties that the participation constraint binds for the lowest type and the utility reduction in the event of a loss, $D(\theta)$ being non-negative and bounded above by $D_0:= v(w) - v(w-\ell) $ for $G$-almost every $\theta$ \citep[Theorem 1]{chade2012optimal}. Observe that $D_0$ measures the drop in utility from a loss when there is no insurance.
Denoting the space of continuous real-valued functions on $\Theta$ by $\mathscr{C}(\Theta)$ (equipped with the sup-norm), it is therefore without loss to restrict attention to the following set of indirect utilities:
\begin{align} \label{eq:indirect_utility_insurance}
    A:=\left\{a\in \mathscr{C}(\Theta): 
    \begin{array}{l}
    a(\underline{\theta})=\underline{\theta} v(w-\ell)+ (1-\underline{\theta})v(w),\\ a \text{ is weakly decreasing, convex 
    and $D_0-$Lipschitz} 
    \end{array}
    \right\}.
\end{align}

The insurer's expected profit from a menu $(D, a)$ chosen by a type-$\theta$ insuree is 
\[u(D,a,\theta):= w-\theta \ell - (1-\theta)v^{-1}(\underbrace{a(\theta)+ \theta D(\theta)}_{\text{no loss}}) - \theta v^{-1}(\underbrace{a(\theta)-(1-\theta)D(\theta)}_{\text{loss}}),\]
where $w - \theta \ell$ is type-$\theta$ insuree's total wealth in expectation, and the remaining terms are the insurer's expected cost of providing the utilities promised to the insuree. To write the insurer's payoff as a function of $a$ alone, we proceed with the following step.
Since $a'(\theta)=-D(\theta)$ almost everywhere, the insurer's payoff is determined by $a$ almost everywhere. At any $\theta$ where $a$ is not differentiable, because $a$ is convex and $D_0$-Lipschitz, the insurer chooses $D(\theta) \in \partial a(\theta) \cap [0,D_0]$, where $\partial a(\theta)$ is the subdifferential of $a$ at $\theta$. Therefore, we can interpret the indirect utility $a$ as the insurer's action: the insuree's payoff from action $a\in A$ is $u_S(a,\theta)=a(\theta)$, and the insurer's payoff from action $a$ is $u_R(a,\theta)= \max_{D\in \partial a(\theta)\cap[0,D_0]} u(D,a,\theta)$.
\footnote{Because $\partial a(\theta)$ is closed for every $\theta$ and $u(D,a, \theta)$ is continuous, $u_R(a,\theta)$ is well-defined.} In the appendix, we verify that $A$ is compact, $u_S(a,\theta)$ is continuous in $a$ for each $\theta$, and $u_R(a,\theta)$ is upper semicontinuous in $(a,\theta)$. We show that the sender's risk aversion induces an ``almost'' strict concavity in the receiver's payoff function, which allows us to establish that the form of continuity described in \Cref{r:assn_insurance} holds.
With this reformulation, the following conclusion holds. 

\begin{proposition} \label{Proposition-Insurance}
    For every achievable payoff profile $(u^*_S,u^*_R)$ and every $\e>0$, there is an equilibrium of the disclosure game that supports payoffs within $\e$ of $(u^*_S,u^*_R)$.
\end{proposition}
\Cref{Proposition-Insurance} speaks to the regulatory debate on the implications of allowing insurance companies to condition their contracts based on genetic tests. If the insuree can partially disclose her test results, then allowing disclosure could potentially lead to significant efficiency gains or even approximate the insuree-optimal information structure. However, the insuree could also be compelled to disclose all evidence, enabling the insurer to capture all efficiency gains. In light of this equilibrium multiplicity, our results emphasize the role that various parties---government agencies, intermediaries, and insurance companies themselves---can play in coordinating behavior towards the public interest.

\section{Conclusion}\label{Section-Conclusion}

This paper evaluates the full potential of hard information in market and contractual settings in which the receiver can flexibly adjust prices, transfers, and allocations in response to the sender’s disclosure. Such settings are canonical in that monopoly pricing, bilateral trade with interdependent values, insurance contracting, and policy negotiations all share this structure. In these settings, full revelation typically leaves the sender worst off: once her type is known, the receiver can tailor the terms of trade to extract her entire surplus. We take this structural feature as our primary departure from the classical disclosure literature. Combined with two standard assumptions, it delivers a sharp equivalence result: the set of equilibrium payoffs in the disclosure game is virtually identical to the set achievable through information design.

Conceptually, this conclusion re-frames the perceived gap between information design and voluntary disclosure. Information design endows the sender with the full power to commit to an information structure. By contrast, voluntary disclosure is typically viewed as a commitment problem, a self-defeating trap in which the sender is forced to reveal everything. Our central finding shows that this gap disappears in many economically important settings: the sender can use hard information to attain her commitment payoff without having to commit. 

From an applied perspective, the findings highlight that hard information not only substitutes for commitment but also equips the sender with a powerful tool for counteracting the receiver’s bargaining advantage. In sender-optimal equilibria, she discloses just enough information to soften the receiver's response while withholding the rest to prevent full surplus extraction. This force highlights the value of intermediaries who can verify statements and coordinate disclosure. Such intermediaries would enable buyers and insurees to capture nearly all the gains available under optimal information design. 

\begin{singlespace}
	\addcontentsline{toc}{section}{References}
	\bibliographystyle{jpe}
	\bibliography{segmentation.bib}
\end{singlespace} 

\newpage

\begin{appendix}
\section*{Main Appendix}
The Main Appendix is organized as follows: \Cref{Appendix-Thm1} proves \Cref{Theorem-MainResult}, \Cref{Appendix-Finite} proves \Cref{Theorem-Finite}, and \Cref{Appendix-PaperApplications} collects our proofs for monopoly pricing and disclosures about asset quality, as described in \Cref{Section-Monopoly,Section-Asset} respectively.

\section{Proof of \autoref{Theorem-MainResult}}\label{Appendix-Thm1}
\subsection{Proof of \autoref{Lemma-FinitePartitionalEqm}}
\begin{proof}[\unskip\nopunct]
    Fix a finite partitional segmentation $\sigma$. Let $(u^*_S,u^*_R)$ denote the payoff pair that it achieves and $a_{\sigma}:\supp \sigma\rightarrow A$ be the receiver's best responses that result in these payoffs. 
    Because $F$ has full support, $\bigcup_{G \in \supp(\sigma)}\supp G = \Theta$, as otherwise the segments could not average to $F$. 
    
    Now consider the following messaging strategy of the sender: If $\theta \in \supp(G)$ and $\theta \notin \supp(H)$ for any $H \in \supp(\sigma)$ such that $H \ne G$, the type-$\theta$ sender sends message $\supp(G)$ with probability 1. If $\theta \in \supp(G)$ for multiple $G \in \supp(\sigma)$, type $\theta$ sends message $\supp H$ with probability 1, where $H\in\argmax_{\{G\in \supp \sigma \, : \, \theta\in \supp G\}}u_S(a_{\sigma}(G),\theta)$. Because $\sigma$ is finite partitional, such types are contained in a $F$-null set and hence do not affect expected payoffs. 
    
    The receiver's belief system is such that whenever she observes message $\supp(G)$ for some $G \in \supp(\sigma)$, she updates using Bayes rule and her new belief is $G$; following any other message $m$, her belief is a point mass at the worst-case type $\hat{\theta}_m$, as defined in \autoref{Assumption-WorstCase}. We specify the receiver's strategy as follows: if she observes message $\supp(G)$ for any $G \in \supp(\sigma)$, she chooses action $a_{\sigma}(G)$ (which is optimal for the receiver given belief $G$); for any other message $m$, she chooses action $\underline{a}(\hat{\theta}_m)$. 

    By construction, the receiver is choosing a best response given her beliefs after any message and beliefs satisfy Bayes' rule whenever possible. Moreover, the sender never wants to deviate from her messaging strategy:  by messaging according to the strategy described above, his payoff is at least $u_S(\underline{a}(\theta),\theta)$ by \autoref{Assumption-NoRents}. If he deviates to an off-path message $m$, his payoff is $u_S(\underline{a}(\hat{\theta}_m),\theta)$, which is lower by \autoref{Assumption-WorstCase}. By construction, those types that can send multiple on-path messages choose optimally among feasible on-path messages.
    Hence, these strategies and beliefs form an equilibrium and this equilibrium induces the segmentation $\sigma$ and payoffs $(u^*_S,u^*_R)$. \end{proof}

\subsection{A Preliminary Step for \autoref{Lemma-FinitePartitionalSuff}}

\begin{lemma}\label{l:approximate_multidimensional}
Let $\sigma$ be a finite segmentation with $\supp(\sigma) = \{F_1, \ldots, F_N\}$. There is a sequence of finite partitional segmentations $\{\sigma^m\}_{m\in \mathbb N}$ with $\supp(\sigma^m)=\{ F^m_1,\ldots,F^m_N \}$ such that, for $i=1,\ldots,N$, $\sigma^m(F^m_i)\rightarrow \sigma(F_i)$ and $F^m_i\rightarrow F_i$.
\end{lemma}
\begin{proof}
    Without loss of generality, suppose $\Theta=[0,1]^n$. For any $m,k\in\mathbb{N}$, partition $\Theta$ into at least $m$ equal-sized cubes; denote this partition by $\mathcal P^m$. Partition each $P\in\mathcal P^m$ further into at least $k$ equal-sized cubes, and denote the collection of all such small cubes by $\Q$. Choose an assignment of all smaller cubes to $\{1,\ldots,N\}$, denoted by $\ell:\Q\rightarrow \{1,\ldots,N\}$, to minimize 
    \begin{align}\label{eq:partition}
        \max_{i\in\{1,\ldots,N\}}\sum_{P\in\mathcal P^m} \left| F\left(\bigcup_{Q\in \Q, Q\subseteq P, \ell(Q)=i}Q\right) - \sigma(F_i)F_i(P)\right|. 
    \end{align} 
    Note that this expression goes to zero as $k\rightarrow \infty$. Therefore, for each $m$, there exists $k(m)\in\mathbb{N}$ such that we can partition each cube into $k(m)$ smaller cubes such that \eqref{eq:partition} is at most $1/m$. Denote this partition by $\Q^m$.
    
    Define a finite partitional segmentation by setting, for all measurable $B\subseteq \Theta$,
    \[ F^m_i(B):=\frac{F(B\cap \bigcup_{Q\in\Q^m, \ell(Q)=i} Q)}{F(\bigcup_{Q\in\Q^m, \ell(Q)=i} Q)} \hspace{0.1in} \text{ and } \hspace{0.1in}  \sigma^m(F^m_i):= F\left(\bigcup_{Q\in\Q^m, \ell(Q)=i} Q\right).    \]
    It follows that $\sigma^m(F^m_i)\rightarrow \sigma(F_i)$ and $\sum_{P\in \mathcal P^m}|F_i^m(P)-F_i(P)|\rightarrow 0$ as $m\rightarrow \infty$.\footnote{Indeed, $\sum_{P\in \mathcal P^m}|F_i^m(P)-F_i(P)|=\frac{1}{\sigma(F_i)}\sum_\mathcal{P}|\frac{\sigma(F_i)}{F(\bigcup_{Q\in\Q^m, \ell(Q)=i} Q)}F(P\cap \bigcup_{Q\in\Q^m, \ell(Q)=i} Q)-\sigma(F_i)F_i(P)|\rightarrow 0$ by definition of $\ell$ (cf. \eqref{eq:partition}) and because $\frac{\sigma(F_i)}{F\left(\bigcup_{Q\in\Q^m, \ell(Q)=i}Q\right)}\rightarrow 1$.}
    Therefore, for any Lipschitz-continuous function $h:\Theta\rightarrow \R$ and $i\in\{1,\ldots, N\}$, $\int h \,\mathrm dF^m_i\rightarrow \int h\,\mathrm dF_i$. Hence, $F^m_i\rightarrow F_i$.
\end{proof}

\subsection{Proof of \autoref{Lemma-FinitePartitionalSuff}}
\begin{proof}[\unskip\nopunct]
    Fix arbitrary $\e>0$ and an arbitrary segmentation $\sigma$ achieving payoffs $(u^*_S,u^*_R)$. 
\bigskip
    
\noindent\textbf{Step 1:} We can assume that the receiver plays $\overline{a}(G)$ for each $G\in\supp \sigma$. \medskip

Indeed, the arguments are symmetric if the receiver plays $\underline{a}(G)$ for each $G\in\supp \sigma$, and the conclusion then follows for arbitrary best responses because the sender's payoff is a convex combination of the payoffs achieved from always playing $\overline{a}$ and always playing $\underline{a}$.
\bigskip

\noindent\textbf{Step 2:} We choose $\e'$ and $\delta$ small enough and define a measurable mapping that maps any segment $G$ to a closeby segment $\overline{H}$ such that all best responses in segment $\overline{H}$ are within $\delta$ of $\overline{a}(G)$. \medskip

Fix $\e'>0$ and $\delta>0$ small enough.  Let $t:\Delta(\Theta)\rightarrow \Delta(\Theta)$ be a measurable function that sends any segment $G$ to $t(G)$, where $\frac{\mathrm dt(G)}{\mathrm d G}\le 1+\e'$ and any best response given segment $t(G)$ is in $B_{\delta}(\overline{a}(G))$. \autoref{Assumption-Continuity}\ref{a:continuity_2} ensures that such a function exists.\bigskip

\noindent\textbf{Step 3:} We define $\sigma_1$ by considering a (scaled-down) version of the pushforward measure of $\sigma$ under $t$ and adding a mass point at an extra segment containing the remaining types. \medskip

For any (measurable) $Z\subseteq \Delta(\Theta)$, define $\tilde \sigma(Z):= \frac{1}{1+\e'}\sigma(t^{-1}(Z))$ to be a (scaled-down) pushforward of $\sigma$ under $t$.

For any (measurable) $E\subseteq \Theta$, $(t(G))(E)= \int_E \frac{\mathrm dt(G)}{\mathrm dG} \,\mathrm d G \le (1+\e') G(E) $ by the bound on the Radon-Nikodym derivative.
Therefore, 
\[ \left(\int H \,\mathrm d\tilde \sigma(H)\right)(E) =  \left(\frac{1}{1+\e'} \int t(G) \,\mathrm d \sigma(G)\right)(E) \le \frac{1}{1+\e'} \int (1+\e') G(E) \,\mathrm d \sigma(G) = F(E) \]
where the first equality follows from the definition of $\tilde \sigma$ and the last equality follows since $\int G \,\mathrm d\sigma(G)=F$. Intuitively, $\tilde \sigma$ does not ``exhaust'' all types available under the prior (and isn't even a probability distribution over segments).
Therefore, we add a segment $G_{extra}$ that contains all the remaining types: Let $G_{extra}:= \frac{F-\int H\,\mathrm  d\tilde \sigma(H)}{1-\tilde\sigma(\Delta(\Theta))}$ and note that $\left(F-\int H \,\mathrm d\tilde \sigma(H)\right)(\Theta)=1-\frac{1}{1+\e'}=1-\tilde\sigma(\Delta(\Theta))$. Hence, $G_{extra}\in\Delta(\Theta)$.
Now define $\sigma_1:= \tilde\sigma + \delta_{G_{extra}} (1-\tilde\sigma(\Delta(\Theta)))$, where $\delta_{G_{extra}}$ is a Dirac measure at $G_{extra}$.
Simple accounting shows that $\sigma_1\in \Delta(\Delta(\Theta))$ and $\int G \,\mathrm d \sigma_1(G) = F$.
\bigskip

\noindent\textbf{Step 4:} We argue that the payoffs under $\sigma_1$ are within $\e/3$ of $(u^*_S,u^*_R)$.\medskip

A fraction $\frac{1}{1+\e'}$ of types end up in a segment under $\sigma_1$ in which any best response for the receiver is within $\delta$ of the best response the receiver would have chosen under segmentation $\sigma$. Since the sender's payoff is continuous in the action for each type, we can choose $\e'>0$ and $\delta>0$ in Step 2 small enough such that the expected payoff for the sender is within $\e/3$ of $u^*_S$.
Similar arguments apply to the receiver's payoff: We can choose $\e'>0$ small enough such that, for any segment $G$, $G$ and $t(G)$ are close in the Levy-Prokhorov metric (which metricizes the weak$^*$-topology), and hence the resulting payoffs for the receiver are close by \autoref{Assumption-Continuity}\ref{a:continuity_1}. Moreover, by choosing $\e'>0$ small enough, the extra segment $G_{extra}$ gets arbitrarily small probability under $\sigma_1$, and hence the receiver's payoff is within $\e/3$ of $u^*_R$.
\bigskip

\noindent\textbf{Step 5:} We define a new segmentation $\sigma_2$, which is an approximation of the segmentation $\sigma_1$, such that $\sigma_2$ contains finitely many segments and such that payoffs under $\sigma_2$ are within $\e/3$ of the payoffs under $\sigma_1$.\medskip

Fix $\delta_2>0$.
For any $G\in\supp \tilde\sigma$, there is an $\e_G$-ball (in the Levy-Prokhorov metric) around $G$ such that any best response to any $G'\in B_{\e_G}(G)$ is within $\delta_2$ of any best response given $G$ (by \autoref{Assumption-Continuity}\ref{a:continuity_1}). Since $\supp \tilde\sigma$ is compact, finitely many of these balls, say $\{B_1,...,B_m\}$, cover $\supp \tilde\sigma$, where $B_i:=B_{\e_{G_i}}(G_i)$. We form a new segmentation $\sigma_2$ by merging all segments that lie in a given ball. Formally, $\sigma_2$ has at most $m+1$ segments in its support; 
the first segment is the barycenter of the set $B_1$ and, for $i>1$,
the $i$th segment is defined recursively as the barycenter of the set $B_i\setminus\bigcup_{j=1}^{i-1}B_j$ under $\tilde\sigma$,
\[H_i:=\frac{\int_{B_i\setminus \bigcup_{j=1}^{i-1}B_j} G \,\mathrm d\tilde \sigma(G)}{\tilde\sigma(B_i\setminus \bigcup_{j=1}^{i-1}B_j)}\]
whenever $\tilde\sigma(B_i\setminus \bigcup_{j=1}^{i-1}B_j)>0$ and we define $\sigma_2(\{H_i\}):=\tilde\sigma(B_i\setminus \bigcup_{j=1}^{i-1}B_j)$ and $\sigma_2(\{G_{extra}\}):=\sigma_1(\{G_{extra}\})$.\footnote{For simplicity, we assume $H_i\neq H_j$ for $i\neq j$ and $H_i\neq G_{extra}$. Our arguments apply without this assumption.}
 One can verify that open balls in the Levy-Prokhorov metric are convex. Therefore, $H_i\in B_i$ and any best response to the merged segment $H_i$ is within $2\delta_2$ of any best response to any $G\in B_i$. 
 By choosing $\delta_2$ small enough, we obtain a segmentation with finite support such that payoffs under $\sigma_2$ are within $\e/3$ of the payoffs under $\sigma_1$.
\bigskip

\noindent\textbf{Step 6:} We define a new segmentation $\sigma_3$ which is finite partitional and approximates $\sigma_2$. We argue that the payoffs under $\sigma_3$ are within $\e/3$ of the payoffs under $\sigma_2$.\medskip

By \autoref{l:approximate_multidimensional}, for any $\delta'>0$ we can approximate the segmentation $\sigma_2$ by a finite partitional segmentation $\sigma_3$ such that to any segment $G$ in $\sigma_2$ there is a unique corresponding segment $H$ with $|\sigma_3(H)-\sigma_2(G)|<\delta'$ and $d_P(G,H)<\delta'$, where $d_P$ denotes the Levy-Prokhorov metric. By \autoref{Assumption-Continuity}\ref{a:continuity_1} we can choose $\delta'$ small enough so that the receiver's payoff under segmentation $\sigma_3$ is within $\e/3$ of the payoff under segmentation $\sigma_2$. Similarly, by choosing $\delta'$ small enough, any optimal action given belief $H$ is close to any optimal action under $G$ and hence the sender's expected payoff given segmentation $\sigma_3$ is within $\e/3$ of the expected payoff given segmentation $\sigma_2$.

It follows that payoffs under $\sigma_3$ are within $\e$ of $(u^*_S,u^*_R)$.
\end{proof}

\begin{remark}\label{r:assn_insurance}
    The proof of \autoref{Lemma-FinitePartitionalSuff} still goes through if we replace \autoref{Assumption-Continuity}\ref{a:continuity_2} by the following alternative assumption:
    \begin{quote}
        For any $G \in\supp \sigma$ and any $\e, \delta > 0$, there is a distribution $\underline{H}$ ($\overline{H}$) whose Radon-Nikodym derivative satisfies $\frac{\mathrm d\underline{H}}{\mathrm dG}\le 1+\e$ ($\frac{\mathrm d\overline{H}}{\mathrm dG}\le 1+\e$) and any best response $a'$ to $\underline{H}$ ($\overline{H}$) is such that $|u_S(a', \theta) - u_S(\underline{a}(G),\theta)| < \delta$ ($|u_S(a', \theta) - u_S(\overline{a}(G),\theta)| < \delta$) for $\underline{H}$($\overline{H}$)-almost every $\theta$. Moreover, the functions that send $G$ to $\overline{H}$ and $\underline{H}$ are measurable.
    \end{quote}
    Since $\supp \underline{H}\subseteq \supp G$, any action optimal under $\underline{H}$ yields almost the same payoff for types in $\supp \underline{H}$. Thus, by choosing sufficiently small $\e', \delta > 0$ in Step 2, Step 4 still goes through. 
\end{remark}

\section{Proof of \Cref{Theorem-Finite}}\label{Appendix-Finite}

We first prove the following preliminary lemma. Without loss of generality, we assume that $\Theta$ is unidimensional.

\newcommand{\uq}{\underline{q}}
\begin{lemma}\label{Lemma-discrete_approx}
For every $N \in \mathbb{N}$, $\uq > 0$, and $\delta > 0$, there exists $\gamma > 0$ such that the following holds: If $F$ has finite support with $F(\{\theta\}) < \gamma$ for all $\theta$, and $\sigma$ is a finite segmentation with $\supp(\sigma) = \{F_1, \ldots, F_n\}$ where $n \le N$ and $\sigma(F_i) \ge \uq$ for all $i$, then there exists a finite partitional segmentation $\bar{\sigma}$ with $\supp(\bar{\sigma}) = \{\bar{F}_1, \ldots, \bar{F}_n\}$ such that for each $i = 1, \ldots, n$:
\[
|\bar{\sigma}(\bar{F}_i) - \sigma(F_i)| < \delta \quad \text{and} \quad d_P(\bar{F}_i, F_i) < \delta,
\]
where $d_P$ denotes the L\'{e}vy-Prokhorov metric.
\end{lemma}
\begin{proof}
Fix $N \in \mathbb{N}$, $\uq > 0$, and $\delta > 0$. Define $\gamma := \frac{\uq \delta}{N + 1}$ and 
suppose $F$ has finite support with $F(\{\theta\}) < \gamma$ for all $\theta$. Let $\sigma$ be any segmentation with $n \le N$ segments satisfying $\sigma(F_i) \ge \uq$ for all $i$. We construct the finite partitional segmentation $\bar\sigma$.

 Let $J := \supp F = \{\theta_1, \ldots, \theta_M\}$, where for any $s, t \in \{1, \ldots, M\}$ with $s<t$, $\theta_s < \theta_t$. Consider an assignment rule $\ell: J \to \{1, \ldots, n\}$ defined as follows: $\ell(\theta_1) = 1$, and for every $s>1$, define iteratively\footnote{To ease exposition, we slightly abuse notation by letting $G(\theta) := G([\uv, \theta])$ for all $\theta \in \Theta$ for any probability measure $G$ with support on $J$, where $\uv := \inf \Theta$.}
    \[\ell(\theta_s) := \min\{i \in \{1, \ldots, n\}: F(X^{s-1}_{i}) \le \sigma(F_i) F_i(\theta_{s-1})\},\]
    where $X^{s-1}_i = \{\theta \in \{\theta_1, \ldots, \theta_{s-1}\}: \ell(\theta) = i\}$.
    For every $i = 1, \ldots, n$, define $J_i := \{\theta \in J: \ell(\theta) = i\}$; $\{J_1, \ldots, J_n\}$ is a partition of $J$. 
    Define a finite partitional segmentation $\bar{\sigma}$ with $\supp(\bar{\sigma})=\{\bar{F}_1,\ldots,\bar{F}_n\}$ by setting, for every measurable $B \subseteq \Theta$, 
    $\bar{F}_i(B):=\frac{F\left(B\cap J_i\right)}{F\left(J_i\right)}$,
    and $\bar{\sigma}(\bar{F}_i):= F\left(J_i\right)$ for each $i=1, \ldots, n$.
    Since $F(\{\theta\}) < \gamma$ for all $\theta \in J$, we obtain
    \begin{equation} \label{eq:disc_1}
        \bar{\sigma}(\bar{F}_i)\bar{F}_i(\theta) \le \sigma(F_i) F_i(\theta)+\gamma
    \end{equation}
    for all $\theta \in \Theta$ and $i=1, \ldots, n$. Because both $\sigma$ and $\bar{\sigma}$ are segmentations, $\sum_{i} \bar{\sigma}(\bar{F}_i) \bar{F}_i(\theta) = F(\theta)= \sum_i \sigma(F_i) {F}_i(\theta)$, and hence
    \[ \sum_{i=1}^n \sigma(F_i) F_i(\theta) = \bar{\sigma}(\bar{F}_k) \bar{F}_k(\theta) + \sum_{i\neq k} \bar{\sigma}(\bar{F}_i) \bar{F}_i(\theta) \le \bar{\sigma}(\bar{F}_k) \bar{F}_k(\theta) +\sum_{i\neq k} \left(\sigma(F_i) F_i(\theta)+\gamma\right)\]
    for all $\theta \in J$ and $k=1, \ldots, n$, where the inequality follows from \eqref{eq:disc_1}. Consequently,
    \begin{equation} \label{eq:disc_2}
        \bar{\sigma}(\bar{F}_i)\bar{F}_i(\theta) \ge \sigma(F_i) F_i(\theta) - (n-1)\gamma.
    \end{equation}
    
    Taking $\theta = \theta_M$, inequalities \eqref{eq:disc_1} and \eqref{eq:disc_2} together imply $\bar{\sigma}(\bar{F}_i) - \sigma(F_i) \le \gamma$ and $\sigma(F_i) - \bar{\sigma}(\bar{F}_i) \le (n-1)\gamma$ for each $i=1,\ldots,n$, respectively. Furthermore, dividing by $\sigma(F_i) \ge \uq$ on both sides of \eqref{eq:disc_1} and rearranging, we get
    \[\bar{F}_i(\theta) - F_i(\theta) \le \frac{\gamma}{\sigma(F_i)} + \frac{\sigma(F_i) - \bar{\sigma}(\bar{F}_i)}{\sigma(F_i)} \bar{F}_i(\theta) \le \frac{\gamma}{\uq} + \frac{(n-1)\gamma}{\uq} = \frac{n\gamma}{\uq}\]
    for all $\theta \in \Theta$. Similarly, dividing by $\sigma(F_i)$ on both sides of \eqref{eq:disc_2}, we obtain $F_i(\theta) - \bar{F}_i(\theta) \le n\gamma/\uq$ for all $\theta \in \Theta$. 
    
    Our choice of $\gamma$ assures that 
    \[|\bar{\sigma}(\bar{F}_i) - \sigma(F_i)| \le (n-1)\gamma \le (N-1) \cdot \frac{\uq\delta}{N+1} < \delta,\]
    and 
    \[\sup_{\theta \in \Theta} |\bar{F}_i(\theta) - F_i(\theta)| \le \frac{n\gamma}{\uq} \le \frac{N}{\uq} \cdot \frac{\uq\delta}{N+1} = \frac{N\delta}{N+1} < \delta.\]
    By definition of the L\'{e}vy-Prokhorov metric, $d_P(\bar{F}_i, F_i) < \delta$. 
\end{proof}

\begin{proof}[Proof of \autoref{Theorem-Finite}]
Under \Cref{Assumption-NoRents,Assumption-WorstCase,Assumption-Continuity}, the proof of \autoref{Lemma-FinitePartitionalEqm} and the first five steps in the proof of \autoref{Lemma-FinitePartitionalSuff} go through without any assumption on the prior $F$. Fix $\e > 0$. We show there exists $\gamma > 0$ such that if $F$ has finite support with $F(\{\theta\}) \le \gamma$ for all $\theta$, then for every finite segmentation $\sigma$, there exists a finite partitional segmentation $\bar{\sigma}$ with payoffs within $\e/3$ of those under $\sigma$.

\medskip
\noindent\textit{Step 1: There exists a segmentation $\hat{\sigma}$ with at most $K$ segments, each with probability at least $\uq$, and with payoff change from $\sigma$ of at most $\e/6$.}

Given any segmentation $\sigma$ with segments $\{F_1, \ldots, F_m\}$, we construct a new segmentation as follows: Using \Cref{Assumption-Continuity} and compactness of $\Delta(\Theta)$, we can cover $\Delta(\Theta)$ by at most $K$ balls $\{B_1,\ldots,B_K\}$ such that if we merge any segments in $\sigma$ that lie in a given ball into a single segment under $\hat{\sigma}$---using the notion of merging as defined in Step 5 of the proof of \Cref{Lemma-FinitePartitionalSuff}---payoffs of the sender and receiver under $\hat{\sigma}$ differ by at most $\e/12$ from the payoffs under $\sigma$. Note that $K$ depends only on $\e$. Moreover, we can choose $\uq>0$ small enough (depending only on $\e$) such that by merging all segments with probability at most $\uq$ into a single segment, payoffs change by at most $\e/12$.

\medskip
\noindent\textit{Step 2: We apply \autoref{Lemma-discrete_approx} to show that there is a finite partitional segmentation $\bar{\sigma}$ and with payoff change from $\hat{\sigma}$ of at most $\e/6$.}

    By \autoref{Lemma-discrete_approx}, for every $\delta > 0$ we can find a $\gamma > 0$ such that as long as $F$ has finite support with $F(\{\theta\})\le \gamma$, there is a finite partitional segmentation $\bar{\sigma}$ such that to any segment $G$ in $\hat{\sigma}$ there is a unique corresponding segment $H$ with $|\bar{\sigma}(H)-\hat{\sigma}(G)|<\delta$ and $d_P(G,H)<\delta$. By \autoref{Assumption-Continuity}\ref{a:continuity_1} we can choose $\delta$ small enough so that the receiver's payoff under segmentation $\bar{\sigma}$ is within $\e/6$ of the payoff under segmentation $\hat{\sigma}$. Similarly, by choosing $\delta$ small enough, any optimal action given belief $H$ is close to any optimal action under $G$ and hence the sender's expected payoff given segmentation $\bar{\sigma}$ is within $\e/6$ of the expected payoff given segmentation $\sigma$. Therefore, the payoffs under $\Bar{\sigma}$ are within $\e/6$ of the payoffs under $\hat{\sigma}$. 
\end{proof}

\section{Proofs for \autoref{Section-Applications}}\label{Appendix-PaperApplications}
\subsection{Proof of \autoref{Proposition-BBM}}

\begin{lemma} \label{l:mp_vb_assn3}
    Let $w: [\uv,\ov] \to \mathbb{R}$ be a continuous and strictly increasing function. If $A$ is an interval of real numbers and the receiver's payoff is given by $u_R(a,\theta) = w(a) \mathbf{1}_{a \le \theta}$, then \autoref{Assumption-Continuity} is satisfied.
\end{lemma}

\begin{proof}
    We first verify \autoref{Assumption-Continuity}\ref{a:continuity_1}. Recall that $u_R(a,G)=\int u_R(a,\theta)\,\mathrm dG(\theta)$. Let $U_R(G):=\max_{a\in A} u_R(a,G)$ be the receiver's optimal payoff in segment $G$. The proof of Theorem 1 in \cite*{yang2023continuity} can be used \emph{mutatis mutandis} to show that the receiver's optimal payoff $U_R(G)$ is continuous under the weak$^*$ topology, and the best response correspondence $\argmax_{a \in A} u_R(a, G)$ is upper hemicontinuous. Thus, \autoref{Assumption-Continuity}\ref{a:continuity_1} is satisfied.

    Verifying \autoref{Assumption-Continuity}\ref{a:continuity_2} requires more work. Let $\underline{a}: \Delta(\Theta) \to \R$ and $\overline{a}: \Delta(\Theta) \to \R$ denote the mappings that map a segment to the lowest optimal action and the highest optimal action in that segment, respectively. Because the best response correspondence is upper hemicontinuous, $\underline{a}$ is lower semicontinuous, and $\overline{a}$ is upper semicontinuous. Given $\e, \delta > 0$, we construct two functions, $\underline{t}: \Delta (\Theta) \to \Delta (\Theta)$ and $\overline{t}: \Delta (\Theta) \to \Delta (\Theta)$, such that both $\underline{t}$ and $\overline{t}$ are measurable, the Radon-Nikodym derivatives satisfy $\frac{\mathrm{d}\underline{t}(G)}{\mathrm{d}G} \le 1+\e$ and $\frac{\mathrm{d}\overline{t}(G)}{\mathrm{d}G} \le 1+\e$ for every $G \in \Delta(\Theta)$, and any optimal action in segment $\underline{t}(G)$ ($\overline{t}(G)$) is contained in $(\underline{a}(G)-\delta, \underline{a}(G)+\delta)$ ($(\overline{a}(G)-\delta, \overline{a}(G)+\delta)$, respectively). For notational ease, we sometimes suppress the dependence of $\underline{a}$ and $\overline{a}$ on $G$ and simply write $\underline{a}$ and $\overline{a}$ for $\underline{a}(G)$ and $\overline{a}(G)$, respectively.

    Define the function $\underline{t}:\Delta(\Theta)\rightarrow \Delta(\Theta)$ by
    \[ \underline{t}(G)(x):= \frac{\min\{ G(x), 1-\kappa \eta(G) \}}{1-\kappa \eta(G)} \]
    where $\eta(G)$ solves 
    \begin{align}\label{eq:eta}
             \frac{\eta(G)}{1-\eta(G)} 2w(\overline{\theta})= \max_{a\in [\underline{a}(G)-\delta/2,\underline{a}(G)]} w(a) [1-G(a)] -\max_{a\le \underline{a}(G)-\delta } w(a) [1-G(a)],         
    \end{align}
    and $\kappa\in(0,1)$ is small enough so that $\frac{d\underline{t}(G)}{dG} \le 1+\e$. Note that $\eta(G)$ is well-defined and takes values in $(0,1)$ because the right-hand side is strictly positive and bounded. 

    We now establish that $\underline{t}$ is measurable. First, $\underline{a}$ is measurable because it is lower semicontinuous. To see that $\eta$ is measurable, reformulate the first maximization problem on the RHS of \eqref{eq:eta} as choosing $(a,q)$ with $(a,q)\in \Gamma(G)$, where
    \[\Gamma(G) := \left\{(a, q) \in \left[\underline{a}(G) - \delta/2, \underline{a}(G)\right] \times [0,1]: 1-G(a) \le q \le 1-G(a-)\right\}\]
    to maximize $w(a)q$. The objective function is continuous in $a$ and $q$, $\Gamma$ has nonempty compact values, and by Lemma 1 in \cite*{yang2023continuity}, $\Gamma$ is continuous and hence weakly measurable. Reformulate the second maximization problem on the RHS of \eqref{eq:eta} analogously. Then we can apply the measurable maximum theorem \citep[Theorem 18.19 in][]{aliprantis2006infinite} to conclude that the RHS of \eqref{eq:eta} is measurable in $G$. This implies that $\eta$ is measurable and, therefore, $\underline{t}$ is measurable.

    Finally, we verify that any optimal action in segment $\underline{t}(G)$ is within $\delta$ of $\underline{a}(G)$: 
    \begin{align*}
        \max_{a\in[\underline{a}(G)-\delta/2,\underline{a}(G)]} w(a)[1-\underline{t}(G)(a)]
        \ge & \max_{a\in[\underline{a}(G)-\delta/2,\underline{a}(G)]} w(a)\left[1-\frac{G(a)}{1-\kappa \eta(G)}\right]\\
        = &\max_{a\in[\underline{a}(G)-\delta/2,\underline{a}(G)]} w(a)\left[1-G(a)-\frac{\kappa \eta(G)}{1-\kappa \eta(G)}G(a)\right]\\
        \ge &\max_{a\in[\underline{a}(G)-\delta/2,\underline{a}(G)]} w(a)[1-G(a)]-w(\overline \theta) \frac{\kappa \eta(G)}{1- \kappa \eta(G)}\\
        > &\max_{a\le \underline{a}(G)-\delta } w(a) [1-G(a)]    \ge \max_{a\le \underline{a}(G)-\delta } w(a) [1-\underline{t}(G)(a)],
    \end{align*}
    where the strict inequality follows from \eqref{eq:eta}.
    Hence, any optimal action in segment $\underline{t}(G)$ is at least $\underline{a}(G)-\delta$. Moreover, for any $a>\underline{a}(G)$, we have
    \begin{align*}
        w\left(\underline{a}\right)\left[1-\frac{G(\underline{a}(G))}{1-\kappa \eta(G)}\right] \ge     w(a)[1-G(a)]-\frac{\kappa \eta(G)}{1-\kappa \eta(G)}w\left(\underline{a}\right)G(\underline{a}) >w(a)\left[1-\frac{G(a)}{1-\kappa \eta(G)}\right].
    \end{align*}
    Hence, no action strictly above $\underline{a}(G)$ is optimal in segment $\underline{t}(G)$. Thus, the function $\underline{t}$ has the desired properties.

Define the function $\overline{t}:\Delta(\Theta)\rightarrow \Delta(\Theta)$ by
\[\overline{t}(G)(x) \equiv \frac{(1-\kappa) G(x)}{1-\kappa G(\overline{a})} \mathbf{1}_{x < \overline{a}(G)} + \frac{G(x) - \kappa G(\overline{a})}{1-\kappa G(\overline{a})} \mathbf{1}_{x \ge \overline{a}(G)},\]
where $\kappa > 0$ is small enough such that $\frac{\mathrm{d}\overline{t}(G)}{\mathrm{d}G} < 1+ \e$. Because $\overline{a}$ is upper semicontinuous, it is measurable; it is straightforward to see that $\overline{t}$ is also measurable.

Next, we show that for every segment $G$, $\overline{a}(G)$ is the unique optimal action in segment $\overline{t}(G)$. For any $a' > \overline{a}(G)$,
\begin{align*}
    w\left(\overline{a}\right)(1-\overline{t}(G)(\overline{a})) & =  w\left(\overline{a}\right) \left[1 - \frac{G(\overline{a}) - \kappa G(\overline{a})}{1-\kappa G(\overline{a})}\right]  = \frac{w\left(\overline{a}\right)(1-G(\overline{a}))}{1-\kappa G(\overline{a})} \\
    & > \frac{w\left(a'\right)(1-G(a'))}{1-\kappa G(\overline{a})}  = w\left(a'\right)(1-\overline{t}(G)(a')),
\end{align*}
where the strict inequality follows because $a' > \overline{a}$, $w$ is strictly increasing, and $\overline{a}$ is the \emph{highest} optimal action. For any $a'' < \overline{a}(G)$,
\begin{align*}
    w\left(\overline{a}\right)(1-\overline{t}(G)(\overline{a})) & = w\left(\overline{a}\right) \frac{1-G(\overline{a})}{1-\kappa G(\overline{a})} \\
    & \ge \frac{(1-\kappa) w\left(a''\right)(1-G(a'')) + \kappa w\left(\overline{a}\right) (1-G(\overline{a}))}{1-\kappa G(\overline{a})} \\
    & > \frac{(1-\kappa) w\left(a''\right)(1-G(a'')) + \kappa w\left(a''\right) (1-G(\overline{a}))}{1-\kappa G(\overline{a})} \\
    & = w\left(a''\right)\left[1-\frac{(1-\kappa)G(a'')}{1-\kappa G(\overline{a})}\right] = w\left(a''\right)(1-\overline{t}(G)(a'')),
\end{align*}
where the weak inequality holds because $\overline{a}$ is an optimal action, and the strict inequality follows since $a'' < \overline{a}(G)$ and $w$ is strictly increasing. Hence, $\overline{t}$ also has the desired properties, and \autoref{Assumption-Continuity}\ref{a:continuity_2} is verified.
\end{proof}

\begin{proof}[Proof of \autoref{Proposition-BBM}]
    It suffices to verify Assumptions \ref{Assumption-NoRents}, \ref{Assumption-WorstCase}, and \ref{Assumption-Continuity}. Because $\underline{a}(\theta) = \theta$, $u_S(\underline{a}(\theta),\theta) = 0$. Since $u_S$ is bounded below by zero, \autoref{Assumption-NoRents} is satisfied. For \autoref{Assumption-WorstCase}, we claim that for every message $m \in \mathcal{C}$, one can set $\hat{\theta}_{m}:=\max_{\theta \in m} \theta$. This is because for every $\theta \in m$, $u_S(\underline{a}(\theta),\theta) = 0 = \max\{\theta - \hat{\theta}_m, 0\} = u_{S}\left(\underline{a}\left(\hat{\theta}_{m}\right), \theta\right)$. Finally, \autoref{Assumption-Continuity} holds since we can appeal to \autoref{l:mp_vb_assn3} by letting $w$ be the identity function. 
\end{proof}

\subsection{Proof of \autoref{Proposition-Interval}}
\begin{proof}[\unskip\nopunct]
If the aggregate market is unimprovable, the fully concealing equilibrium---in which every type sends the message $[\uv,\ov]$---supports the payoffs of the consumer-optimal segmentation.

For the converse, suppose the aggregate market is improvable and let $p^*>\underline{\theta}$ denote the lowest uniform monopoly price for the aggregate market. Towards a contradiction, suppose that there is an interval-message equilibrium that yields ex ante consumer payoffs of at least $w^*-\pi^*-1/n$ for every $n\in\mathbb{N}$. The monopolist's payoff in the $n^{th}$ equilibrium is therefore at most $\pi^*+1/n$ and the welfare loss from exclusion is at most $1/n$.

Consequently, for any $m>0$ and all $n$ sufficiently large, in the segmentation induced by the $n^{th}$ equilibrium, there is a set of segments $\mathcal G_n^m$ whose probability is at least $1-1/m$ and such that each segment $G\in\mathcal G_n^m$ satisfies: (i) the monopolist's optimal profit in segment $G$ exceeds the profit from charging $p^*$ by at most $1/m$, and (ii) at most a fraction $1/m$ of types in segment $G$ is excluded. Property (i) implies that for all $n$ sufficiently large, each $G \in\mathcal G_n^m$ contains types in $[p^*,p^*+1/m]$. Property (ii) implies that for $n$ sufficiently large, the equilibrium price in each $G\in\mathcal G_n^m$ is at most $G^{-1}(1/m)$. It follows that for $m$ and $n$ large enough, $\mathcal G_n^m$ contains a segment with equilibrium price close to $\underline{\theta}$.

In equilibrium, each type sends whichever available message yields the lowest price, provided this price is below their valuation. Consequently, for any $\e>0$, by properties (i) and (ii) it holds that for all $m$ and $n$ sufficiently large, types in $[\underline{\theta}+\e,p^*]$ will send messages that result in the same price and at least one of these messages contains a type that is at least $p^*+1/m$ (otherwise, uniformly lowering the price from $p^*$ in all these segments would increase profit, contradicting that $p^*$ is close to optimal in all segments). But this implies that all types $\theta\in[p^*,p^*+1/m]$ will in equilibrium send a message resulting in a price close to $\underline \theta$. By property (i), all segments in $\mathcal{G}_n^m$ must therefore result in this low price. But as $m,n\rightarrow \infty$, the probability of $\mathcal G_n^m$ converges to 1, and because $\underline{\theta}$ is not an optimal uniform price, equilibrium profits fall strictly below $\pi^*$, yielding a contradiction.
\end{proof}

\subsection{Proof of \autoref{Proposition-TruthLeaning}}
\begin{proof}[\unskip\nopunct]
    To establish \ref{p:truthleaningpart1}, it suffices to show that trade must happen with probability 1 in any truth-leaning equilibrium. In any equilibrium of the perturbed game $\Gamma^{\boldsymbol{\e}}$, trade must happen with probability 1: any type $\theta$ that does not trade must fully reveal her type to get payoff $\boldsymbol{\e}(\theta) > 0$, and the unique optimal action for the monopolist after receiving message $\{\theta\}$ with $\theta>0$ is $a=\theta$, resulting in trade. By definition, a truth-leaning equilibrium is a limit point of equilibria of $\Gamma^{\boldsymbol{\e}}$, and hence trade also happens with probability 1 in any such equilibrium.

    To establish \ref{p:truthleaningpart2}, fix $\e > 0$ and an efficient payoff profile $(u^*_S, u^*_R)$. By \autoref{Proposition-BBM}, there exists a payoff profile $(u^1_S, u^1_R)$ that is within $\e/2$ of $(u^*_S, u^*_R)$ and that is supported by an equilibrium $e^*$. Moreover, the proof of \autoref{Theorem-MainResult} shows that we can assume that $e^*$ induces a finite partitional segmentation and each segment has positive probability.

    We modify $e^*$ to obtain an efficient equilibrium $e^{**}$ of the unperturbed game that will also be an equilibrium of some perturbed games: In $e^{**}$, all consumer types whose payoff is strictly positive in $e^*$ send the same message as in $e^*$; all types whose payoff is zero in  $e^*$ send the fully revealing messages. The monopolist's strategy is as in $e^*$ and beliefs are derived from Bayes' rule whenever possible (with skeptical beliefs after off-path messages).

    To verify that $e^{**}$ is an equilibrium of the unperturbed game, note that each consumer type gets the same payoff in $e^{**}$ as in $e^*$, any deviation to an on-path message would yield the same payoff as in $e^*$, and any deviation to an off-path message is not profitable. Hence, the consumer best responds. Moreover, the monopolist's actions are still optimal: Consider an on-path message $m$ that is sent with positive probability, and denote the segment and price induced by message $m$ in $e^*$ by $G$ and $p_G$, respectively. Under the modified strategy of the consumer in $e^{**}$, the segment induced by $m$ is $G(\,\cdot \, | \, \theta > p_G) := (G(\theta)-G(p_G))/(1-G(p_G))$. The monopolist's profits from charging $p$ under $G(\,\cdot \, | \, \theta > p_G)$ is therefore
    \[\tilde{\Pi}(p) = p (1-G(p \,|\,\theta > p_G)) = \frac{p(1-G(p))}{1-G(p_G)}.\]
    Since $p_G$ is a maximizer of $p(1-G(p))$, it also maximizes $\tilde{\Pi}(p)$ among all $p \in [p_G, \ov]$. Therefore, the monopolist best responds to $m$. For any other message, the monopolist holds skeptical beliefs and best responds as well.

    Denote by $(u^2_S, u^2_R)$ the payoff profile induced by $e^{**}$. To see that $(u^2_S, u^2_R)$ is within $\e$ of $(u^*_S, u^*_R)$, note that because $(u^*_S,u^*_R)$ and $(u^2_S, u^2_R)$ are efficient and the efficiency frontier has slope $-1$, 
    \[|u^*_R - u^2_R| = |u^*_S - u^2_S| = |u^*_S - u^1_S| < \frac{\e}{2}.\]
    
    It remains to argue that the equilibrium $e^{**}$ is truth-leaning. Towards this end, for every $n\in\mathbb N$ and for every $\theta$ that does not send a fully revealing message in $e^{**}$, define $\boldsymbol{\e}^n(\theta) = (\theta-p(m))/2n$, where $m$ is the message sent by type $\theta$ in $e^{**}$ and $p(m)$ the resulting price, and for every type $\theta$ that sends a fully revealing message in $e^{**}$, define $\boldsymbol{\e}^n(\theta) = \ov/n$. Then, for each $n$, $\boldsymbol{\e}^n(\theta)>0$ for all $\theta$ and $e^{**}$ is an equilibrium of the perturbed game $\Gamma^{\boldsymbol{\e}^n}$ since no type of the consumer has a profitable deviation, and $\boldsymbol{\e}^n$ converges uniformly to $\mathbf{0}$. Therefore, $e^{**}$ constitutes a truth-leaning equilibrium, which completes the proof of \ref{p:truthleaningpart2}.\end{proof}

\subsection{Proof of \autoref{Proposition-Akerlof}}

\begin{proof}[\unskip\nopunct]
     \Cref{Assumption-NoRents,Assumption-WorstCase} can be verified in a manner similar to the proof of \Cref{Proposition-BBM}. To verify \autoref{Assumption-Continuity}\ref{a:continuity_1}, fix an arbitrary $G$ and a sequence $G_n\rightarrow G$. 
    We show that 
    \begin{align}
        &\forall a\in A:\forall a_n\rightarrow a: \ & \limsup_n u_R(a_n,G_n)&\le u_R(a,G) \text{ and } \label{eq:limsup}\\
        &\forall \e>0:\forall a\in A: \exists a_n\rightarrow a: \ &  \liminf_n u_R(a_n,G_n)&\ge u_R(a,G)-\e. \label{eq:liminf}
    \end{align}
This will imply that $u_R(\cdot, G_n)$ $\Gamma$-converges to $u_R(\cdot,G)$ \citep[page 288 in][]{santambrogio2023}.\footnote{More precisely, it will imply that $-u_R(\cdot,G_n)$ $\Gamma$-converges to $-u_R(\cdot,G)$. Since we are interested in maximization problems instead of minimization problems, these are the relevant inequalities for us.} Propositions 7.4 and 7.5 in \cite{santambrogio2023} then imply that $U_R(G_n)\rightarrow U_R(G)$ (hence $U_R(G)$ is continuous) and the set of optimal actions is upper hemicontinuous.

To verify \eqref{eq:limsup}, observe that 
\[ \limsup_n u_R(a_n,G_n)-u_R(a,G) \le \limsup_n [u_R(a_n,G_n)-u_R(a,G_n)] + \limsup_n [u_R(a,G_n) - u_R(a,G)],\]
the first $\limsup$ on the RHS is weakly negative because $u_R(\cdot,\theta)$ is upper semicontinuous and the second $\limsup$ on the RHS is weakly negative by the Portmanteau theorem because $u_R(a,\theta)$ is upper semicontinuous in $\theta$ (since $v(\theta)>c(\theta)$) and $G_n\rightarrow G$.

To verify \eqref{eq:liminf}, define $c^{-1}(a):=\max\{\theta\in \Theta: c(\theta)\le a\}$ and fix an arbitrary $a$. Without loss of generality, assume $G(c^{-1}(a))>0$ and choose an arbitrary $\e\in(0,G(c^{-1}(a)))$. Note that there is a sequence $a_n\rightarrow a$ satisfying $a_n>a$ and $G_n(c^{-1}(a_n))\ge G(c^{-1}(a))-\e$ for all $n$ because $G_n\rightarrow G$.  Then
\begin{align*}
    u_R(a,G)
    &=  \int \mathbf(v(\theta)-a) \,\mathrm d[\min\{G(\theta),G(c^{-1}(a))\}]\\
    &= \lim_n \int \mathbf(v(\theta)-a) \,\mathrm d[\min\{G_n(\theta),G(c^{-1}(a))\}]\\
    &\le \lim \int \mathbf(v(\theta)-a) \,\mathrm d[\min\{G_n(\theta),G(c^{-1}(a))-\e\}] +\e v(\overline{\theta})\\
    &\le \lim \int \mathbf{1}_{c(\theta)\le a_n}\mathbf(v(\theta)-a) \,\mathrm d[\min\{G_n(\theta),G(c^{-1}(a))-\e\}] +\e v(\overline{\theta})\\ 
    &\le \lim \int \mathbf{1}_{c(\theta)\le a_n}\mathbf(v(\theta)-a) \,\mathrm dG_n(\theta) +\e v(\overline{\theta})\\
    &= \lim u_R(a_n,G_n) +\e v(\overline{\theta}) 
\end{align*}
where the second equality holds because $G_n\rightarrow G$ and the integrand is continuous, the second inequality holds because $G_n(c^{-1}(a_n))\ge G(c^{-1}(a))-\e$, the third inequality holds because the integrand is weakly positive, and the last equality holds because $a_n\rightarrow a$.
Given that $\e$ is arbitrary, \eqref{eq:liminf} follows.

It remains to verify \autoref{Assumption-Continuity}\ref{a:continuity_2}. Given $\e, \delta > 0$, we construct two functions $\underline{t}: \Delta (\Theta) \to \Delta (\Theta)$ and $\overline{t}: \Delta (\Theta) \to \Delta (\Theta)$ that satisfy \autoref{Assumption-Continuity}\ref{a:continuity_2}.
Define the function $\underline{t}:\Delta(\Theta)\rightarrow \Delta(\Theta)$ by \[\underline{t}(G)(x) = \frac{G(x)}{G(c^{-1}(\underline{a})) + 1 - \kappa} \mathbf{1}_{x \le c^{-1}(\underline{a}(G))} + \frac{\kappa G(c^{-1}(\underline{a})) + (1-\kappa) G(x)}{G(c^{-1}(\underline{a})) + 1-\kappa} \mathbf{1}_{x > c^{-1}(\underline{a}(G))},\]
where $\kappa \in (0,1)$ is small enough such that $\frac{\mathrm{d}\overline{t}(G)}{\mathrm{d}G} < 1+ \e$. It is not difficult to see that $\underline{t}$ is measurable; we next show that for every segment $G$, $\underline{a}(G)$ is the unique optimal price in segment $\underline{t}(G)$. For any $a' < \underline{a}(G)$,
\begin{align*}
    & \int_{\uv}^{c^{-1}(\underline{a})} (v(\theta) - \underline{a}) \, \mathrm{d}\underline{t}(G)(\theta) = \frac{\int_{\uv}^{c^{-1}(\underline{a})} (v(\theta) - \underline{a}) \, \mathrm{d}G(\theta)}{G(c^{-1}(\underline{a})) + 1-\kappa} \\ > ~ & \frac{\int_{\uv}^{c^{-1}(a')} (v(\theta) - a') \, \mathrm{d}G(\theta)}{G(c^{-1}(\underline{a})) + 1-\kappa} = \int_{\uv}^{c^{-1}(a')} (v(\theta) - a') \, \mathrm{d}\underline{t}(G)(\theta), 
\end{align*}
where the strict inequality holds because $\underline{a}$ is the lowest optimal price for belief $G$. Moreover, for any $a'' > \underline{a}(G)$,
\begin{align*}
    & \int_{\uv}^{c^{-1}(a'')} (v(\theta) - a'') \, \mathrm{d}\underline{t}(G)(\theta) \\ = ~ &
    \frac{1}{G(c^{-1}(\underline{a})) + 1-\kappa} \left[(1-\kappa) \int_{\uv}^{c^{-1}(a'')} (v(\theta) - a'') \, \mathrm{d}G(\theta) + \kappa \int_{\uv}^{c^{-1}(\underline{a})} (v(\theta) - a'') \, \mathrm{d}G(\theta)\right] \\ \le ~ &
    \frac{1}{G(c^{-1}(\underline{a})) + 1-\kappa} \left[(1-\kappa) \int_{\uv}^{c^{-1}(\underline{a})} (v(\theta) - \underline{a}) \, \mathrm{d}G(\theta) + \kappa \int_{\uv}^{c^{-1}(\underline{a})} (v(\theta) - a'') \, \mathrm{d}G(\theta)\right] \\ < ~ &
    \frac{1}{G(c^{-1}(\underline{a})) + 1-\kappa} \int_{\uv}^{c^{-1}(\underline{a})} (v(\theta) - \underline{a}) \, \mathrm{d}G(\theta) = \int_{\uv}^{c^{-1}(\underline{a})} (v(\theta) - \underline{a}) \, \mathrm{d}\underline{t}(G)(\theta),
\end{align*}
where the weak inequality holds because $\underline{a}$ is an optimal price for belief $G$, and the strict inequality follows since $a'' > \underline{a}(G)$. Thus, the function $\underline{t}$ has the desired properties.

Define the function $\overline{t}:\Delta(\Theta)\rightarrow \Delta(\Theta)$ by
\[\overline{t}(G)(x):= \max\left\{0, \frac{G(x)-\kappa \eta(G)}{1-\kappa \eta(G)}\right\},\]
where $\eta(G)$ solves 
    \begin{align}\label{eq:eta2}
             \frac{\eta(G)}{1-\eta(G)} 2v(\overline{\theta})= \max_{a\in [\overline{a}(G), \overline{a}(G)+\delta/2]} u_R(a,G) -\max_{a\ge \overline{a}(G)+\delta } u_R(a,G),         
    \end{align}
    and $\kappa\in(0,1)$ is a small number so that $\frac{d\overline{t}(G)}{dG} \le 1+\e$. Note that $\eta(G)$ is well-defined and takes values in $(0,1)$ because the right-hand side is strictly positive and bounded. Because $u_R(a,G) = \int (v(\theta)-a)\,\mathrm{d}G_a(\theta)$, where $G_a(\theta) := \min\{G(\theta), G(c^{-1}(a))\}$, is a Carathéodory function, and both $\Gamma_1(G) := [\overline{a}(G), \overline{a}(G)+\delta/2]$ and $\Gamma_2(G) := [\overline{a}(G)+ \delta,\ov]$ are weakly measurable correspondences, $\eta$ is measurable by the measurable maximum theorem. Thus, $\overline{t}$ is also measurable. 

We verify that any optimal action in segment $\overline{t}(G)$ is within $\delta$ of $\overline{a}(G)$: First, 
\begin{align*}
    \max_{a\in[\overline{a}(G), \overline{a}(G)+\delta/2]} \int_{\uv}^{c^{-1}(a)}(v(\theta) - a)\, \mathrm{d}\overline{t}(G)(\theta)
    & \ge \max_{a\in[\overline{a}(G), \overline{a}(G)+\delta/2]} \int_{\uv}^{c^{-1}(a)}(v(\theta) - a)\, \mathrm{d} G(\theta)  \\
    & >  \max_{a \ge \overline{a}(G)+\delta} \int_{\uv}^{c^{-1}(a)}(v(\theta) - a)\, \mathrm{d} G(\theta) +  \frac{ \eta(G)}{1- \eta(G)} v(\overline{\theta})\\
    & \ge \max_{a \ge \overline{a}(G)+\delta} \int_{\uv}^{c^{-1}(a)}(v(\theta) - a)\, \mathrm{d} \overline{t}(G)(\theta),
\end{align*}
where the first inequality follows because $\overline{t}(G)$ first-order stochastically dominates $G$, the strict inequality follows from \eqref{eq:eta2}, and the last inequality holds as long as $\kappa$ is small enough.
Second, for any $a < \overline{a}(G)$, we have
\begin{align*}
    & \int_{\uv}^{c^{-1}(\overline{a}(G))} (v(\theta) - \overline{a}(G)) \, \mathrm{d}\left(\frac{G(\theta) - \kappa \eta(G)}{1 - \kappa \eta(G)}\right) \\ \ge ~ & \int_{\uv}^{c^{-1}(a)} (v(\theta) - a) \, \mathrm{d}G(\theta) - \frac{\kappa \eta(G)}{1 - \kappa \eta(G)} \int_{\uv}^{c^{-1}(a)} (v(\theta) - \overline{a}(G)) \, \mathrm{d}\left(1-G(\theta)\right) \\
    \ge ~ & \int_{\uv}^{c^{-1}(a)} (v(\theta) - a) \, \mathrm{d}\left(\frac{G(\theta) - \kappa \eta(G)}{1 - \kappa \eta(G)}\right) - (\overline{a}(G) - a) G(c^{-1}(a)) \\ >~ & \int_{\uv}^{c^{-1}(a)} (v(\theta) - a) \, \mathrm{d}\left(\frac{G(\theta) - \kappa \eta(G)}{1 - \kappa \eta(G)}\right).
\end{align*}
Hence, no price that has distance at least $\delta$ from $\overline{a}(G)$ can be optimal when the belief is $\overline{t}(G)$.
Thus, the function $\overline{t}$ has the desired properties.
\end{proof}

\subsection{Proof of \autoref{Proposition-Auctions}}
\iffalse
\begin{proof}[\unskip\nopunct]
It is well known that a game in which two or more buyers compete in price is equivalent to an auxiliary game in which the only difference is that a single buyer always offers the expected value of the product based on disclosure. Let $((\rho,\tau),\mu)$ be an equilibrium of the auxiliary game; we show that for every $\theta \in \Theta$, the type-$\theta$ seller's equilibrium payoff is $v(\theta)$. Fix an arbitrary $\theta \in \Theta$, the equilibrium condition \ref{eq_cond_a} implies that $u_S(\tau(m),\theta) \ge v(\theta)$ for every $m \in \supp \rho(\theta)$, as otherwise the seller can profitably deviate to the fully revealing message. It suffices to show that the inequality must hold with equality for every $m \in \supp \rho(\theta)$. Suppose the inequality is strict, then it must be that $\tau(m) > v(\theta)$. This implies that $m \in \supp \rho(\hat{\theta})$ for some $\hat{\theta} > \theta$ with $v(\hat{\theta}) > \tau(m)$, which violates the equilibrium condition \ref{eq_cond_a} for $\hat{\theta}$, since type $\hat{\theta}$'s payoff from a fully revealing message is $v(\hat{\theta})$. Therefore, for every type $\theta \in \Theta$, the seller's equilibrium payoff is $v(\theta)$, and hence the equilibrium $((\rho,\tau),\mu)$ is outcome equivalent to a fully revealing equilibrium.
\end{proof}
\fi

\begin{proof}[\unskip\nopunct]
With two or more buyers engaged in Bertrand competition for the asset, every equilibrium involves  the highest offer for the asset equaling its expected value conditional on disclosure. Consider any belief $G$ of the receiver induced in an equilibrium, and define $\ov_G := \max \supp G$. We argue that the expected value of the asset under belief $G$, $\E_G[v(\theta)]$, must satisfy $v(\ov_G) - \e \le \E_G[v(\theta)] \le v(\ov_G)$ for every $\e > 0$. The first inequality holds because otherwise some seller type that is inducing belief $G$ would find it profitable to deviate to sending the fully revealing message. The second inequality holds by definition. Given that $\e$ is arbitrary, we obtain $\E_G[v(\theta)] = v(\ov_G)$. As $v$ is strictly increasing, $G$ must then assign probability 1 to type $\theta_G$. Therefore, every equilibrium belief is degenerate, which implies that every equilibrium is outcome equivalent to a fully revealing equilibrium.
\end{proof}

\newpage

\section{Supplementary Appendix}

This appendix is organized as follows:
\begin{itemize}[noitemsep]
    \item \Cref{Section-IntuitiveCriterion} proves that equilibria we consider satisfy the Intuitive Criterion. 
    \item \Cref{Section-Dye} proves robustness of our main result in principal-agent settings in which the sender lacks evidence with positive probability (``Dye evidence''). 
    \item \Cref{Section-Failure} describes how the main conclusion of \Cref{Theorem-MainResult} fails when any one of \Cref{Assumption-NoRents,Assumption-WorstCase,Assumption-Continuity} is dropped.
    \item \Cref{Section-InsuranceProofs} proves \Cref{Proposition-Insurance}.
    \item \Cref{Section-VB} describes our application to policy negotiations. 

\end{itemize}

\subsection{Intuitive Criterion}\label{Section-IntuitiveCriterion}
	\begin{theorem}\label{Theorem-IntuitiveCriterion}
		Every equilibrium of the game survives the Intuitive Criterion.
	\end{theorem} 

    \begin{proof}
        Fix an arbitrary equilibrium and let $v_S^*(\theta)$ denote type $\theta$ sender's interim payoff in this equilibrium. For a belief $\mu$, denote by $BR(\mu) \coloneqq \argmax_{a \in A} \int u_R(a,\theta)\, \mathrm{d}\mu(\theta)$ the receiver's best response correspondence. For any measurable $I\subseteq \Theta$, define $BR(I)$ as the set of responses that are optimal for the receiver under some beliefs whose support is contained in $I$: $BR(I) \coloneqq \bigcup_{\mu: \supp \mu \subseteq I} BR(\mu)$. Finally, for any off-path message $m$, let
        \[S(m) \coloneqq \left\{\theta \in m: v_S^*(\theta) > \max_{a \in BR(m)} u_S(a,\theta)\right\};\]
        that is, $S(m)$ consists of the types in $m$ whose equilibrium payoff exceeds the highest payoff attainable when the receiver's belief is concentrated on $m$.

        For any off-path message $m$ with $S(m) \ne m$, and any $\theta \in m \setminus S(m)$, we have
        \[v^*_S(\theta) \ge u_S(\underline{a}(\theta),\theta) \ge \min_{a \in BR(m \setminus S(m))} u_S(a,\theta),\]
        where the first inequality follows from the fact that every type gets at least her complete-information payoff in any equilibrium, and the second inequality holds because the Dirac measure on $\theta$, $\delta_{\theta}$, satisfies $BR(\delta_{\theta}) \subseteq BR(m \setminus S(m))$.
        This implies that the equilibrium $((\rho, \tau), \mu)$ survives the Intuitive Criterion. 
    \end{proof}

\subsection{Robustness to Dye Evidence}\label{Section-Dye}

Herein, we restrict attention to principal-agent settings corresponding to \Cref{Example-PrincipalAgent}. We perturb our game as follows. With probability $\zeta \in (0,1)$, the sender does not possess any evidence to disclose regardless of her type. Therefore, for every type $\theta$, the sender's message space is $\mathcal{M}(\theta)$ with probability $(1-\zeta)$ and is $\{\Theta\}$ with complementary probability. Denote this perturbed game by $\Upsilon^\zeta$. In this perturbed game, with probability $\zeta$, the sender cannot send any evidence and is forced to use the fully-concealing message $\Theta$. This kind of evidence structure is in the spirit of \cite{dye1985disclosure}. Our analysis evaluates the degree to which equilibrium payoffs are robust to the sender lacking evidence to disclose with small probability.

\begin{definition}
	A payoff profile of the original game $\Upsilon^0$ is \textbf{Dye-supportable} if it is a limit point of payoff profiles of equilibria of $\Upsilon^\zeta$ as $\zeta \to 0$. 
\end{definition}

A Dye-supportable payoff profile is one that is supported by, at the limit, equilibria of disclosure games that have the Dye evidence structure. We show below that all achievable payoffs of the principal-agent setting of \Cref{Example-PrincipalAgent} are Dye-supportable, under one additional assumption listed below.  

\begin{assumption}\label{Assumption-PADye}
The following hold:
    \begin{enumerate}[label=(\alph*),noitemsep]
        \item The agent's WTP for every alternative $x \neq x_0$, namely $v_S(x,\theta)-v_S(x_0,\theta)$, strictly increases in her type $\theta$.
        \item The outside option $x_0$ does not maximize $v_S(x,\overline\theta)+v_R(x,\overline\theta)$.
    \end{enumerate}
\end{assumption}
\Cref{Assumption-PADye} is a mild condition. Part (a) holds whenever $X$ is partially ordered with $x_0$ as the lowest alternative and $v_S(x,\theta)$ satisfies strict increasing differences. Part (b) reflects some gains from trade: namely, that the receiver would not choose the outside option if he believed that the sender is type $\overline\theta$ with probability $1$. We observe that both (a) and (b) are satisfied by all three applications considered in \Cref{Section-Applications}. 

\begin{proposition}
Consider a principal-agent setting that satisfies  \Cref{Assumption-Continuity,Assumption-PADye}. Then every achievable payoff profile $(u_S^*,u_R^*)$ is Dye-supportable.
\end{proposition}

\begin{proof}
   	As we argued in the main text, both \Cref{Assumption-NoRents,Assumption-WorstCase} hold in this setting. In particular, every message $m$ contains a worst-case type $\hat{\theta}_m \equiv \max_{\theta \in m} \theta$, and the worst-case type of the fully-concealing message $\Theta$ is $\hat{\theta}_{\Theta} \equiv \overline{\theta}$. 
    
    Fix an achievable payoff profile $(u_S^*, u_R^*)$. For any $\e>0$, \Cref{Theorem-MainResult} delivers an equilibrium $e$ of the original game $\Upsilon^0$ whose induced finite partitional segmentation achieves payoffs within $\e$ of $(u_S^*,u_R^*)$.

	For any $\zeta\in (0,1)$, choose $\xi:=\xi(\zeta)>0$ so that $F\big((\overline\theta-\xi,\overline\theta]\big)=\sqrt{\zeta}$, and define the ``top'' interval $I_\zeta := (\overline\theta-\xi,\overline\theta]$.
	We consider the following strategy profile in the perturbed game $\Upsilon^\zeta$. When the agent has access to all messages in $\mathcal{M}(\theta)$, we prescribe:
	\begin{itemize}
		\item if $\theta\in I_\zeta$, she sends $\Theta$;
		\item if $\theta\notin I_\zeta$, she sends the truncated on-path message $\tilde m(\theta):=
		m(\theta)\cap[\underline\theta,\overline\theta-\xi]$, where $m(\theta)$ is the on-path message sent by type $\theta$ in the fixed equilibrium of the original game.
	\end{itemize}
	We set the principal's response after any off-path message $m$ equal to $\underline{a}(\hat\theta_m)$; after any on-path message, the principal updates using Bayes' rule and best responds. 
	
	Below, we show that the strategy profile described above is an equilibrium of $\Upsilon^\zeta$ for all $\zeta$ sufficiently small. We then show that the sequence of associated equilibrium payoff profiles converges as $\zeta\rightarrow 0$ to the payoff profile induced by equilibrium $e$. Since $\e$ was arbitrary, $(u_S^*,u_R^*)$ is Dye-supportable.\footnote{Formally, take a sequence of equilibria $(e_n)$ whose payoffs converge to $(u_S^*,u_R^*)$. For each $e_n$, there is a sequence of equilibria $e_n^k$ of the perturbed games $\Upsilon^{1/k}$ whose payoffs converge to the payoff of $e_n$. Then the payoffs of the sequence of equilibria $e_n^{n}$ of the perturbed games $\Upsilon^{1/n}$ converges to $(u_S^*,u_R^*)$.}
	\bigskip
	
	\noindent\textbf{Step 1:} We pin down the principal's beliefs upon observing $\Theta$ and his best response as $\zeta\rightarrow 0$.\medskip
	
	\noindent Under the strategy defined above, for every Borel subset $B$ of $\Theta$, Bayes' rule pins down the principal's posterior belief after observing $\Theta$ as
	\begin{equation}\label{eq:posteriorTheta}
		\mu^\zeta_{\Theta}(B)
		=
		\frac{\int_B\big(\zeta+(1-\zeta)\mathbf{1}_{\theta\in I_\zeta}\big)\,
        \mathrm{d}F(\theta)}
		{\int_{\Theta}\big(\zeta+(1-\zeta)\mathbf{1}_{\theta\in I_\zeta}\big)\, \mathrm{d}F(\theta)}
		=
		\frac{\zeta F(B)+(1-\zeta)F(B\cap I_\zeta)}{\zeta+(1-\zeta)F(I_\zeta)}.
	\end{equation}
	Since $F(I_\zeta)=\sqrt{\zeta}$, the weight on $I_\zeta$ in \eqref{eq:posteriorTheta} converges
	to $1$ as $\zeta\to 0$, and $\xi(\zeta)\to 0$; hence $\mu^\zeta_{\Theta} \to \delta_{\overline\theta}$ as $\zeta\to 0$. Then by \Cref{Assumption-Continuity}\ref{a:continuity_1}, for all sufficiently small $\zeta$, every best response to $\mu^\zeta_{\Theta}$ lies in an arbitrarily small Hausdorff-neighborhood	of $a^*(\delta_{\overline\theta})$.
\bigskip
	
	\noindent\textbf{Step 2:} We show that the strategy profile described above is an equilibrium of $\Upsilon^\zeta$ for all sufficiently small $\zeta$.
	\medskip
	
	\noindent As in \Cref{Lemma-FinitePartitionalEqm}, no type profits from deviating to an off-path message, and since $\sigma$
	is finite partitional, almost no type has a profitable deviation to a different on-path message other than $\Theta$. Thus, it suffices to show that no type $\theta\in[\underline\theta,\overline\theta-\xi]$ finds it profitable to deviate to $\Theta$. Let $X^*(\overline\theta)\subseteq X$ denote the set of efficient allocations for type $\overline\theta$, i.e.,
	\[X^*(\overline\theta):=\argmax_{x\in X}\big(v_S(x,\overline\theta)+v_R(x,\overline\theta)\big),\]
	and for $x_{\overline\theta}^*\in X^*(\overline\theta)$ let $t(x_{\overline\theta}^*):=v_S(x_{\overline\theta}^*,\overline\theta)-v_S(x_0,\overline\theta)$. 
    By \autoref{Assumption-PADye},
		\[v_S(x_{\overline\theta}^*,\theta)-v_S(x_0,\theta) < v_S(x_{\overline\theta}^*,\overline\theta)-v_S(x_0,\overline\theta)=t(x_{\overline\theta}^*)\]
		for any $\theta \in [\underline\theta,\overline\theta-\xi]$, which implies that $v_S(x_{\overline\theta}^*,\theta)-t(x_{\overline\theta}^*) < v_S(x_0,\theta)$ for all such types. Moreover, since $X^*(\overline\theta)$ is compact and $v_S(\cdot,\theta)$ is continuous for every $\theta$, we can define
		\[Y(\xi):=\min_{x\in X^*(\overline\theta)}\Big(\big[v_S(x,\overline\theta)-v_S(x_0,\overline\theta)\big]-\big[v_S(x,\overline\theta-\xi)-v_S(x_0,\overline\theta-\xi)\big]\Big)>0;\]
	then for every $x\in X^*(\overline\theta)$ and every $\theta\in[\underline\theta,\overline\theta-\xi]$, $v_S(x,\theta)-t(x)\le v_S(x_0,\theta)-Y(\xi)$.

	Now take any menu $a\in a^*(\delta_{\overline\theta})$, and let $(x_a,t_a)\in a$ be the option selected by type $\overline\theta$ (ties are broken in the principal's favor). Optimality under belief $\delta_{\overline\theta}$ implies $x_a\in X^*(\overline\theta)$ and $v_S(x_a,\overline\theta)-t_a=v_S(x_0,\overline\theta)$, hence $t_a=t(x_a)$. Therefore, any type in $[\underline\theta,\overline\theta-\xi]$ strictly prefers the outside option to $(x_a,t_a)$ by at least $Y(\xi)$, and the same inequality holds for any other $(x,t)\in a$ because type-$\overline\theta$ agent's payoff from $(x,t)$ is at most $\overline\theta$'s outside option payoff. Consequently, for every $a\in a^*(\delta_{\overline\theta})$ and every $\theta\in[\underline\theta,\overline\theta-\xi]$, the outside option is the unique optimal choice from $a$.

	By continuity of the agent's payoff in the menu and the definition of $Y(\xi)$, there exists a Hausdorff-neighborhood $\mathcal E$ of $a^*(\delta_{\overline\theta})$ such that for every menu $a\in\mathcal E$ and every $\theta\in[\underline\theta,\overline\theta-\xi]$, the agent still strictly prefers the outside option and hence obtains exactly $v_S(x_0,\theta)$. Because $a^*(\cdot)$ is upper hemicontinuous by \Cref{Assumption-Continuity}\ref{a:continuity_1} and $\mu^\zeta_{\Theta} \to \delta_{\overline\theta}$, for all sufficiently small $\zeta$ we have $a^*(\mu^\zeta_{\Theta})\subseteq \mathcal E$. Therefore, if the principal best responds to $\Theta$ by choosing any $a^\zeta(\Theta)\in a^*(\mu^\zeta_{\Theta})$, then for every $\theta\in[\underline\theta,\overline\theta-\xi]$, $u_S(a^\zeta(\Theta),\theta)=v_S(x_0,\theta)$.
	Since \Cref{Assumption-NoRents} guarantees that each such type's payoff from using the on-path message is at least $v_S(x_0,\theta)$, no type in $[\underline\theta,\overline\theta-\xi]$ can profitably deviate to $\Theta$. This establishes that the strategy profile defined above is an equilibrium of $\Upsilon^\zeta$ for all sufficiently small $\zeta$.
	\bigskip

    \noindent \textbf{Step 3:} We show that as $\zeta \to 0$, the sequence of payoff profiles associated with the equilibria of $\Upsilon^\zeta$ constructed above converges to the payoff profile of equilibrium $e$.  \medskip

    \noindent We decompose payoffs by considering contributions from types in $I_\zeta$ and types in $[\underline\theta, \overline\theta - \xi]$ separately.

For types in $I_\zeta$, since sender and receiver payoffs are bounded and $F(I_\zeta) = \sqrt{\zeta} \to 0$, the contribution of these types to expected payoffs vanishes as $\zeta \to 0$.

For types in $[\underline\theta, \overline\theta - \xi]$: as $\zeta \to 0$, we have $\xi(\zeta) \to 0$, so the truncated messages $\tilde{m}(\theta) = m(\theta) \cap [\underline\theta, \overline\theta - \xi]$ converge to the original messages $m(\theta)$. Consequently, the segment induced by each truncated on-path message converges to the corresponding segment in equilibrium $e$. By \Cref{Assumption-Continuity}\ref{a:continuity_1}, the receiver's expected payoff from best responding is continuous in beliefs, so the receiver's expected payoff from these types converges to that in equilibrium $e$. By upper hemicontinuity of $a^*(\cdot)$ and compactness of $A$, any sequence of receiver's best responses to the truncated segments converges to a best response in the original segments. Since $u_S(\cdot, \theta)$ is continuous for each $\theta$ and all best responses to a given segment are close-by (by the construction of equilibrium $e$ in \Cref{Lemma-FinitePartitionalSuff} using \Cref{Assumption-Continuity}\ref{a:continuity_2}), the sender's expected payoff from these types also converges to that in equilibrium $e$.

Consequently, the payoff profiles of the constructed equilibria of $\Upsilon^\zeta$ converge to that of equilibrium $e$ as $\zeta \to 0$.
\end{proof}

\subsection{What if the Key Assumptions Fail?}\label{Section-Failure}
Herein, we show that our central conclusion does not hold if any assumption is dropped (while maintaining the other two assumptions).

\begin{example}[\Cref{Assumption-NoRents} fails]
    Consider the following stylized version of \citeauthor{grossman:81-JLE}'s and \citeauthor{milgrom1981good}'s model. The sender's type is uniformly distributed on $\Theta=[0,1]$ and the receiver chooses an action $a$ in $[0,1]$. The sender's motives are transparent in that her payoffs $u_S(a,\theta)$ do not vary with $\theta$ but are strictly increasing and strictly concave in $a$. The receiver would like to match the action with the sender's type and his payoff is $u_R(a, \theta) = -(a-\theta)^2$. \autoref{Assumption-NoRents} fails in this game: the complete information payoff for the sender of type $\theta=1$ is higher than her payoff if the receiver's beliefs equal the prior.

    Here, the unique equilibrium outcome coincides with full revelation. However, given strict concavity, the best achievable payoff for the sender comes from the receiver obtaining no information and choosing action $1/2$. 
\end{example}

\begin{example}[\Cref{Assumption-WorstCase} fails]
    Suppose $A= \{1, 2\}$, the sender's type $\theta$ is uniformly distributed on $\Theta=[0,1]$, the receiver's payoffs are $u_R(1,\theta) = 1-\theta$ and $u_R(2,\theta) = \theta$, and the sender’s payoffs are 
    $u_S(1,\theta) = \left(\frac{1}{2}\right)\mathbf{1}_{\theta = \frac{1}{2}}+  \mathbf{1}_{\theta > \frac{1}{2}}$ 
    and $u_S(2,\theta) = \frac{1}{2}$. \Cref{Assumption-WorstCase} fails in this game, as we show. Observe that the complete-information payoff is $0$ for types strictly below $\frac{1}{2}$ and $\frac{1}{2}$ for all other types. We argue that the message $m=\Theta$ lacks a worst-case type. If the putative ``worst-case'' type $\hat\theta_m$ were assigned to be strictly below $\frac{1}{2}$, then types above $\frac{1}{2}$ accrue more than their complete-information payoff; were $\hat\theta_m$ assigned to be above $\frac{1}{2}$, then types strictly below $\frac{1}{2}$ do better than their complete-information payoff.

    This failure has implications for the equilibrium outcomes of the disclosure game. For instance, no equilibrium supports a payoff profile near that of the fully revealing experiment, $\left(\frac{1}{4},\frac{3}{4}\right)$. Were there such an equilibrium, action $2$ would have to be played with high probability whenever the type is strictly above $\frac{1}{2}$ and action $1$ would have to be played with high probability whenever the type is strictly below $\frac{1}{2}$. However, following the message $m=\Theta$, either action $1$ is played with probability at least $\frac{1}{2}$ and types above $\frac{1}{2}$ could profitably deviate to this message, or action $2$ is played with probability at least $\frac{1}{2}$, in which case types below $\frac{1}{2}$ could profitably deviate. 
\end{example}

\begin{example}[\autoref{Assumption-Continuity} fails]
    Suppose there are $n$ types, $\{\theta_1,\ldots, \theta_n\}$ and $n+1$ actions, $A=\{0,1,\ldots,n\}$.\footnote{For simplicity, we assume that the prior $F$ is supported on a finite set, but this example can be extended to a continuum setting.} The sender's payoff is $u_S(a,\theta)=\mathbf{1}_{a=0}$ and the receiver's payoff is $u_R(a,\theta)=\mathbf{1}_{a=0}+n\mathbf{1}_{\theta=\theta_a}$. 
\autoref{Assumption-Continuity}\ref{a:continuity_2} fails here: every action is optimal under a uniform belief but whenever the belief is not uniform, action 0 is not optimal.

     Given this failure, there are achievable payoff profiles that cannot be (approximately) supported in the disclosure game. Consider a prior that is a convex combination of a uniform belief, with weight $\alpha$, and a point mass at $\theta_1$, with weight $(1-\alpha)$. A segmentation with two segments, one of which is uniform, achieves a payoff of $\alpha$ to the sender. We argue that in every equilibrium, however, the sender's payoff is $0$. To see why, observe that following any message $m\neq \Theta$, the receiver would not choose action $0$. Thus, the only prospect for a strictly positive payoff for the sender is if the receiver played action $0$ with positive probability following the message $m=\Theta$. That cannot happen in an equilibrium: were the receiver to do so, every sender-type would send this message with probability $1$, which would make action $0$ a sub-optimal choice for the receiver. 
\end{example}       
\subsection{Proof of \autoref{Proposition-Insurance} (Insurance Markets)}\label{Section-InsuranceProofs}

\begin{lemma}\label{l:compactness}
    The set $A$ defined in \eqref{eq:indirect_utility_insurance} is compact.
\end{lemma}
\begin{proof}
    Let $\{a_n\}_{n \in \mathbb{N}}$ be a sequence in $A$. Since $a_n$ is uniformly bounded and $D_0-$Lipschitz continuous for each $n$, Arzela-Ascoli's theorem implies there is a subsequence, also denoted by $\{a_n\}$, that converges in the supremum-norm to some $a$. Clearly, $a\in A$. 
\end{proof}

\begin{lemma}\label{l:payoff_continuity_insurance}
    For all $\theta$, $u_S(a,\theta)$ is continuous in $a$.
    Also, $u_R(a,G)$ is upper semicontinuous in $(a,G)$.
\end{lemma}
\begin{proof}
    If $a_n\rightarrow a$ then $u_S(a_n,\theta)=a_n(\theta)\rightarrow a(\theta)=u_S(a,\theta)$ and the first claim follows. 
    
    For the second claim, we show first that $u_R(a,\theta)$ is upper semicontinuous in $(a,\theta)$. Note that if $a_n\rightarrow a$, $\theta_n\rightarrow \theta$, and $\e>0$ then $\partial a_n(\theta_n)\subseteq B_{\e}(\partial a(\theta))$ for all $n$ large enough \citep[Theorem D.6.2.7 in][]{hiriart2004fundamentals}. Hence, $(a,\theta)\mapsto \partial a(\theta)$ is upper hemicontinuous. Since $u(D,a,\theta)$
    is continuous in $(D,a,\theta)$, a maximum theorem \citep[see Lemma 17.30 in][]{aliprantis2006infinite} implies that $u_R(a,\theta)$ is upper semicontinuous in $(a,\theta)$.

    Now consider $a_n\rightarrow a$ and $G_n\rightarrow G$. By Theorem 25.6 in \cite*{billingsley} there are $Y_n$ and $Y$ on a common probability space $(\Omega,\mathcal F,P)$ with distributions $G_n$ and $G$ such that $Y_n(\omega)\rightarrow Y(\omega)$ for all $\omega$. Then
    \begin{align*}
       &\limsup_{n \to \infty} \int u_R(a_n,\theta) \,\mathrm dG_n(\theta) = \limsup_{n \to \infty} \int u_R(a_n,Y_n(\omega)) \,\mathrm dP(\omega) \\
    \le &\int u_R(a,Y(\omega)) \,\mathrm dP(\omega) =  \int u_R(a,\theta) \,\mathrm dG(\theta) 
    \end{align*}
    where the equalities follow from a change of variables and the inequality follows from (reverse) Fatou's lemma and the fact that $u_R(a,\theta)$ is upper semicontinuous in $(a,\theta)$. We conclude that $u_R(a,G)$ is upper semicontinuous in $(a,G)$.
\end{proof}

\begin{lemma}\label{l:differentiable_approx}
    Given $a\in A$, $G\in\Delta(\Theta)$, and $\e>0$ there is $\Tilde{a} \in A$ that is continuously differentiable such that $\norm{a-\Tilde{a}}_\infty<\e/2$ and $|u_R(\Tilde{a},G)-u_R(a,G)| < \e/2$.
\end{lemma}
\begin{proof}
    Fix $a \in A$; let $E'$ denote the set of $\theta$'s at which $a$ is not differentiable. Because $E'$ is countable, we can find a finite set $E'' = \{\theta_1, \ldots, \theta_M\} \subseteq E'$ such that $G(E' \setminus E'') < \e/(12w)$, where 
    $w$ is an upper bound for $|u_R(a,\theta) - u_R(b,\theta)|$ for any $a, b \in A$.
    Let $D^*(\theta) \in \argmax_{D\in \partial a(\theta)\cap[0,D_0]} u(D,a,\theta)$; and for each $i=1,\ldots,M$, let $\ell_i$ denote the supporting hyperplane of $a$ at $\theta_i$ with slope $D^*(\theta_i)$.

    For any $n\in \mathbb N$, we can find a piecewise affine function $\tilde a_n$ such that $\tilde a_n \in A$ and $\norm{\tilde a_n-a}_\infty < 1/(9n)$. Define $\hat a_n:=\max\{\tilde a_n-1/(9n), \ell_1, \ldots, \ell_M\}+1/(9n)$; it can be shown that $\norm{\hat a_n- \tilde a_n}_\infty < 1/(9n)$, and $\hat a_n \in A$ for $n$ large enough. We can approximate $\hat a_n$ by a differentiable function $a_n$ such that $\norm{a_n-\hat a_n}_{\infty} \le 1/(9n)$ and $|a_n'(\theta)-D^*(\theta)|\le 1/(9n)$ for all $\theta\in E''$ \citep[][Theorem 2.26]{rockafellarwets}.

    Since $a_n$ is convex and differentiable, $a_n\rightarrow a$ 
    implies that $a_n'(\theta) \to a'(\theta)$ for all $\theta \notin E'$. By Egoroff's theorem, there exists $E \subseteq \Theta \setminus E'$ with $G(E) < \e/(12w)$ such that $a'_n \rightarrow a$ uniformly on $\Theta \setminus (E' \cup E)$. Because $|a_n'(\theta) - D^*(\theta)| < 1/(9n)$ for all $\theta \in E''$, $|u_R(a_n,\theta) - u_R(a, \theta)| < \e/3$ for any $\theta \not \in (E' \setminus E'')\cup E$ and all $n$ large enough. Then since $G(E' \setminus E'') < \e/(12w)$ and $G(E) < \e/(12w)$, $|u_R(a_n,G) - u_R(a, G)| < \e/2$ for large enough $n$. 
\end{proof}

\begin{lemma}\label{l:insurance_ass_3i}
    The receiver's optimal payoff is continuous in $G$ and the receiver's optimal actions are upper hemicontinuous in $G$.
\end{lemma}
\begin{proof}
Because $u_R(a,G)$ is upper semicontinuous (\autoref{l:payoff_continuity_insurance}), the optimal payoff $U_R(G):=\max_{a\in A} u_R(a,G)$ is upper semicontinuous by Lemma 17.30 in \cite*{aliprantis2006infinite}.

We argue that $U_R(G)$ is also lower semicontinuous: Consider $G_n\rightarrow G$ and suppose towards contradiction there is $\e>0$ such that $\liminf_n U_R(G_n)+\e <U_R(G)$. By \autoref{l:differentiable_approx}, there is a continuously differentiable $\tilde{a}\in A$ such that $u_R(\tilde{a}, G)+\e/2>U_R(G)$. Then $u_R(\tilde{a},\theta)$ is continuous in $\theta$, and therefore $\int u_R(\tilde{a},\theta)\,\mathrm dG_n(\theta)\rightarrow \int u_R(\tilde{a},\theta)\,\mathrm dG(\theta)$ (since we use the weak$^*$ topology). Hence, for all $n$ large enough, $u_R(\tilde{a}, G_n)+\e/2\ge u_R(\tilde{a},G) $. This implies $u_R(\tilde{a}, G_n)+\e\ge U_R(G)$, a contradiction.

Finally, we show that $a^*(G)=\argmax_{a\in A} u_R(a,G)$ is upper hemicontinuous: Consider sequences $G_n\rightarrow G$ and $a_n\in a^*(G_n)$ such that $a_n\rightarrow a$, and suppose towards contradiction that $a\not\in a^*(G)$. Because $a\in A$ this implies $U_R(G)> u_R(a,G)$. Because $u_R(a,G)$ is upper semicontinuous (\autoref{l:payoff_continuity_insurance}),
\[U_R(G)> u_R(a,G) \ge \limsup_{n \to \infty} u_R(a_n,G_n)=\limsup_{n \to \infty} U_R(G_n).\] 
This contradicts lower semicontinuity of $U_R$.
\end{proof}

\begin{proof}[Proof of \autoref{Proposition-Insurance}]
    Observe that for each type $\theta$, the sender's \emph{no-insurance payoff} is $ \theta v(w - \ell) + (1-\theta) v(w)$. \autoref{Assumption-NoRents} holds because if the insuree is known to be of type $\theta$, the insurer will offer a contract that extracts all surplus leaving the insuree with her no-insurance payoff. Given that the no-insurance payoff is strictly decreasing in $\theta$, for any $m \in \mathcal{C}$, the worst-case type is $\max_{\theta \in m} \theta$, which verifies \autoref{Assumption-WorstCase}. \autoref{Assumption-Continuity}\ref{a:continuity_1} follows from \autoref{l:insurance_ass_3i}. Instead of verifying \autoref{Assumption-Continuity}\ref{a:continuity_2}, we verify its alternative in \autoref{r:assn_insurance}. Because  $u_R(a,\theta)$ is strictly concave in $a$ for each $\theta$,\footnote{Take $\lambda \in (0,1)$, any $a', a'' \in A$ with $a' \ne a''$, and any $D' \in \argmax_{D \in \partial a'(\theta) \cap [0,D_0]} u(a',D,\theta)$ and $D'' \in \argmax_{D \in \partial a''(\theta) \cap [0,D_0]} u(a'',D,\theta)$. Then
    \begin{align*}
        \lambda u_R(a', \theta) + (1-\lambda) u_R(a'', \theta) 
        & < w-\theta \ell - (1-\theta)v^{-1}((\lambda a'(\theta) + (1-\lambda) a''(\theta)) + \theta (\lambda D' + (1-\lambda) D'')) \, - \\
        & ~~~~ \theta v^{-1}((\lambda a'(\theta) + (1-\lambda) a''(\theta)) - (1-\theta)(\lambda D' + (1-\lambda) D'')) \\
        & \le u_R(\lambda a' + (1-\lambda) a'', \theta),
    \end{align*}
    where the first inequality holds because $v^{-1}$ is strictly convex (since $v$ is strictly concave), and the second equality follows from the fact that $\lambda D'+(1-\lambda)D'' \in \partial (\lambda a' + (1-\lambda) a'')(\theta) = \lambda \partial  a'(\theta) + (1-\lambda) \partial  a''(\theta)$ for each $\theta$ \citep[Theorem D.4.1.1 in][]{hiriart2004fundamentals}.} for any $G \in \Delta(\Theta)$, any $a', a'' \in a^*(G)$ satisfy $a'=a''$ $G$-almost everywhere. The alternative assumption holds by setting $t(G) = G$ and $\underline{H} = \overline{H} = t(G)$.
\end{proof}

 \subsection{Bargaining Over Policies}\label{Section-VB}

This section applies our analysis to models of policy negotiations with incomplete information that build on \cite{romer1978political}.\footnote{The literature has studied various formulations of veto bargaining with incomplete information; see, for example, \cite{matthews1989veto}, \cite*{kartik2021delegation},  \cite*{ali2023sequential}, and \cite*{kim2025persuasion}.} 
A policy $a\in \Re$ is jointly chosen by the proposer and the vetoer. The proposer's payoff from policy $a$, $u(a)$, is strictly increasing in $a$. The vetoer's payoff, $v(a,\theta)$, is strictly single-peaked in $a$ with a unique maximizer $\theta$ and symmetric around the maximizer; we call the vetoer's ideal policy $\theta$ her \emph{type}. Her type is her private information, and is drawn according to an absolutely continuous CDF $F$ on $\Theta = [\uv,\ov]$ with $\ov<\infty$. For simplicity, we assume $\uv \ge 0$. Once the proposer proposes a policy $a$, the vetoer can accept or reject. If she accepts, the proposed policy prevails; if she rejects, then the status-quo policy $a_Q = 0$ is preserved. Given the proposer's payoffs, she never proposes any $a < 0$. Restricting attention to $a \ge 0$, the vetoer accepts if and only if $a \le 2\theta$. Therefore, we assume that proposals lie in $[0,2\overline{\theta}]$.

We augment this game with a disclosure stage, with the following timing. The vetoer first observes her type $\theta$ and sends a message $m \in \mathcal{M}(\theta)$ to the proposer. The proposer then proposes a policy $a$. The type-$\theta$ vetoer's payoff is given by $u_S(a,\theta) = \max\{v(a,\theta), v(0,\theta)\}$, and the proposer's payoff is given by $u_R(a,\theta) = (u(a)-u(0))\mathbf{1}_{a \le 2\theta}$. For every $\theta \in \Theta$, $u_S(\cdot, \theta)$ is continuous, and $u_R(\cdot, \theta)$ is upper semicontinuous.

We note that this setting does not feature transferable utility and, unlike monopoly pricing, the sender and receiver may have aligned preferences that favor a higher action to $a_Q$. Nevertheless, the complete-information payoff is the lowest for a sender-type $\theta$ because the receiver would then propose action $2\theta$, which results in the same payoff as the status quo. 

In this setting, the implications of information design have not been characterized; nevertheless, the equivalence between design and disclosure holds.

\begin{proposition} \label{p:veto_bargaining}
    For every achievable payoff profile $(u^*_S,u^*_R)$ and every $\e>0$, there is an equilibrium of the disclosure game that supports payoffs within $\e$ of $(u^*_S,u^*_R)$.
\end{proposition}
\begin{proof}
    We verify \Cref{Assumption-NoRents,Assumption-WorstCase,Assumption-Continuity}. Because $\underline{a}(\theta) = 2\theta$ and $v$ is symmetric around $\theta$, $v(\underline{a}(\theta), \theta) = v(0,\theta)$. Consequently, $u_S(\underline{a}(\theta),\theta) = v(0,\theta)$. For every $\theta \in \Theta$, $u_S(a,\theta)$ is bounded below by $v(0,\theta)$, and hence \autoref{Assumption-NoRents} must hold. \autoref{Assumption-WorstCase} also holds because for every message $m \in \mathcal{C}$, one can set $\hat{\theta}_{m}:=\max_{\theta \in m} \theta$. Finally, by re-writing the proposer's payoff as $u_R(b,\theta) = w(b) \mathbf{1}_{b \le \theta}$, where $b=a/2$ and $w(b) := u(2b) - u(0)$, \autoref{l:mp_vb_assn3} implies that \autoref{Assumption-Continuity} holds since $w$ is continuous and strictly increasing by assumption. \end{proof}

\end{appendix}
\end{document}